\documentclass[english,10.5pt]{article}
\usepackage[T1]{fontenc}
\usepackage[latin9]{inputenc}
\usepackage{geometry}
\usepackage{fancyhdr}
\usepackage{color}
\usepackage{mathrsfs}
\usepackage{mathtools}
\usepackage{algorithm2e}
\usepackage{amsmath}
\usepackage{amsthm}
\usepackage{amssymb}
\usepackage{stmaryrd}
\usepackage{graphicx}
\usepackage{setspace}
\usepackage{esint}

\makeatletter

\providecommand{\tabularnewline}{\\}

\newcommand{\lyxaddress}[1]{
\par {\raggedright #1
\vspace{1.4em}
\noindent\par}
}
\theoremstyle{plain}
\newtheorem{thm}{\protect\theoremname}
  \theoremstyle{plain}
  \newtheorem{lem}[thm]{\protect\lemmaname}
  \theoremstyle{definition}
  \newtheorem{defn}[thm]{\protect\definitionname}
  \theoremstyle{remark}
  \newtheorem{conclusion}[thm]{\protect\conclusionname}

\usepackage{amsmath}
\usepackage{amsthm}
\usepackage{amssymb}
\usepackage{hyperref}
\usepackage{graphics}
\usepackage{algorithmic}
\usepackage{subfig}
\usepackage{url}
\usepackage{bm}

\usepackage{nccmath}

\usepackage{setspace}

\usepackage{slashed}

\usepackage{esint}
\usepackage{stmaryrd}

\makeatletter

\usepackage{amsfonts}
\usepackage{ascmac}

\usepackage{framed,color}

\definecolor{shadecolor}{gray}{0.9}

\usepackage{ulem}

\makeatother

\allowdisplaybreaks[4]

\usepackage{feynmp-auto}
\makeatother

\usepackage{babel}
  \providecommand{\conclusionname}{Conclusion}
  \providecommand{\definitionname}{Definition}
  \providecommand{\lemmaname}{Lemma}
\providecommand{\theoremname}{Theorem}

\begin{document}

\title{A Brownian Particle and Fields II:\\
Radiation Reaction as an Application}

\author{{\Large{}Keita Seto}\thanks{keita.seto@eli-np.ro}}

\maketitle

\lyxaddress{\begin{center}
Extreme Light Infrastructure \textendash{} Nuclear Physics (ELI-NP)
/ \\
Horia Hulubei National Institute for R\&D in Physics and Nuclear Engineering
(IFIN-HH), \\
30 Reactorului St., Bucharest-Magurele, jud. Ilfov, P.O.B. MG-6, RO-077125,
Romania.
\par\end{center}}
\begin{abstract}
{\bf Radiation reaction} has been investigated traditionally in classical
dynamics and recently in non-linear QED as {\bf {high-intensity field physics}}
produced by high-intensity lasers. Its quantumness is predicted by
the factor in the radiation formula.  In this {\bf {Volume II}},
the quantization of radiation reaction by using a Brownian scalar
electron is discussed for obtaining the above radiation formula. Finally,
its stochasticity is found as the origin of the factor of its quantumness
in high-intensity field physics regime.
\end{abstract}
$\,$$\,$$\,$$\,$


{[}mathematics{]} Applications of stochastic analysis

$\,$$\,$$\,$$\,$

\maketitle

\thispagestyle{fancy}
\rhead{\texttt{preprint:ELI-NP/RA5-TDR 0004}} 
\lhead{}
\renewcommand{\headrulewidth}{0.0pt}

\setstretch{1.2}

\vfill{}
\newpage{}\thispagestyle{fancy}
\rhead{Keita Seto}
\lhead{ELI-NP/IFIN-HH}
\renewcommand{\headrulewidth}{1.0pt}\tableofcontents

$\,$$\,$$\,$\pagebreak{}

\begin{center}
{\bf Notation and Conventions}
\par\end{center}

\begin{center}
\begin{tabular}{ll}
\hline 
Symbol & Description\tabularnewline
\hline 
\hline 
$c$ & Speed of light\tabularnewline
$\hbar$ & Planck's constant\tabularnewline
$m_{0}$ & Rest mass of an electron\tabularnewline
$e$ & Charge of an electron\tabularnewline
$\mathbb{V_{\mathrm{M}}^{\mathrm{4}}}$ & 4-dimentional standard vector space for the metric affine space\tabularnewline
$g$ & Metric on $\mathbb{V_{\mathrm{M}}^{\mathrm{4}}}$; $g\coloneqq\mathrm{diag}(+1,-1,-1,-1)$\tabularnewline
$\mathbb{A}^{4}(\mathbb{V_{\mathrm{M}}^{\mathrm{4}}},g)$ & 4-dimentional metric affine space with respect to $\mathbb{V_{\mathrm{M}}^{\mathrm{4}}}$
and $g$\tabularnewline
$\mathscr{B}(I)$ & Borel $\sigma$-algebra of a topological space $I$.\tabularnewline
$(\mathbb{A}^{4}(\mathbb{V_{\mathrm{M}}^{\mathrm{4}}},g),\mathscr{B}(\mathbb{A}^{4}(\mathbb{V_{\mathrm{M}}^{\mathrm{4}}},g)),\mu)$ & Minkowski spacetime\tabularnewline
$\varphi_{E}\coloneqq\{\varphi_{E}^{A}|A\in\mathrm{set\,of\,indexes}\}$ & Coordinate mapping on $E$; $\varphi_{E}^{A}:E\rightarrow\mathbb{R}$\tabularnewline
$\left(\mathit{\Omega},\mathcal{F},\mathscr{P}\right)$ & Probability space\tabularnewline
$\mathbb{E}\llbracket\hat{X}(\bullet)\rrbracket\coloneqq\int_{\Omega}d\mathscr{P}(\omega)\,\hat{X}(\omega)$ & Expectation of $\hat{X}(\bullet)\coloneqq\{\hat{X}(\omega)|\omega\in\varOmega\}$\tabularnewline
$\mathbb{E}\llbracket\hat{X}(\bullet)|\mathcal{\mathscr{C}}\rrbracket$ & Conditional expectation of $\hat{X}(\bullet)$ given $\mathcal{\mathscr{C}}\subset\mathcal{F}$\tabularnewline
$\mathcal{\mathscr{P}}_{\tau}\subset\mathcal{F}$  & Sub-$\sigma$-algebra in the family of ''the past'', $\{\mathcal{\mathscr{P}}_{\tau}\}_{\tau\in\mathbb{R}}$\tabularnewline
$\mathscr{F}_{\tau}\subset\mathcal{F}$  & Sub-$\sigma$-algebra in the family of ''the future'', $\{\mathscr{F}_{\tau}\}_{\tau\in\mathbb{R}}$\tabularnewline
$\hat{x}(\circ,\bullet)$ & Dual progressively measurable process (D-progressive, D-process) \tabularnewline
 & {[}Relativistic kinematics of a scalar (spin-less) electron{]}\tabularnewline
 & $\hat{x}(\circ,\bullet)\coloneqq\{\hat{x}(\tau,\omega)\in\mathbb{A}^{4}(\mathbb{V_{\mathrm{M}}^{\mathrm{4}}},g))|\tau\in\mathbb{R},\omega\in\varOmega\}$\tabularnewline
$\mathscr{\mathcal{V}}\in\mathbb{V_{\mathrm{M}}^{\mathrm{4}}}\oplus i\mathbb{V_{\mathrm{M}}^{\mathrm{4}}}$ & Complex velocity; $\mathcal{V}^{\alpha}(x)\coloneqq i\lambda^{2}\times\partial{}^{\alpha}\ln\phi(x)+\begin{gathered}\frac{e}{m_{0}}\end{gathered}
A{}^{\alpha}(x)$\tabularnewline
$\mathcal{F}_{(\mathring{+})},\mathcal{F}_{(\mathring{-})}\in\mathbb{V_{\mathrm{M}}^{\mathrm{4}}}\otimes\mathbb{V_{\mathrm{M}}^{\mathrm{4}}}$ & Retarded field and advanced field\tabularnewline
$\mathfrak{F}\in\mathbb{V}_{\mathrm{M}}^{\mathrm{4}}\otimes\mathbb{V}_{\mathrm{M}}^{\mathrm{4}}$ & Effective radiation (reaction) field\tabularnewline
$\delta\hat{x}(\tau,\omega)\in\mathbb{V_{\mathrm{M}}^{\mathrm{4}}}$ & $\delta\hat{x}(\tau,\omega)\coloneqq\hat{x}(\tau,\omega)-\mathbb{E}\llbracket\hat{x}(\tau,\bullet)\rrbracket$\tabularnewline
\hline 
\end{tabular}
\par\end{center}

\pagebreak{}

\section{Introduction}

This series of the papers proposes quantum dynamics of a stochastic
scalar (spin-less) electron interacting with classical light fields.
It is equivalent to the system of the Klein-Gordon equation and the
Maxwell equation. In this {\bf Volume II}, we focus the topics on
{\bf radiation reaction} as an application of {\bf Volume I} \cite{Seto Volume I}.
It means {\bf the quantization of the Lorentz-Abraham-Dirac (LAD) equation of classical dynamics}.
Radiation reaction as a model of a radiating electron has become important
in the research projects of high-intensity lasers \cite{Mourou(1997a)},
represented by ``Extreme Light Infrastructure (ELI)'' \cite{ELI-NP,ELI-beams,ELI-ALPS}
for the last laser-plasma science \cite{Piazza(2012a)}. Corresponding
to the name of ``high-energy particle physics'', let us name such
a high-intensity laser science {\bf ''high-intensity field physics''}.
Before entering into the main topic, we summarize the available theoretical
works on radiation reaction. 

Denote by $\mathscr{B}(I)$ the Borel $\sigma$-algebra of a topological
space $I$ and a metric affine space $\mathbb{A}^{4}(\mathbb{V_{\mathrm{M}}^{\mathrm{4}}},g)$
with respect to a 4-dimentional standard vector space $\mathbb{V_{\mathrm{M}}^{\mathrm{4}}}$
and its metric $g\coloneqq\mathrm{diag}(+1,-1,-1,-1)$, let a measure
space $(\mathbb{A}^{4}(\mathbb{V_{\mathrm{M}}^{\mathrm{4}}},g),\mathscr{B}(\mathbb{A}^{4}(\mathbb{V_{\mathrm{M}}^{\mathrm{4}}},g)),\mu)$
be the Minkowski spacetime. The coordinate mappings $\varphi_{E}\coloneqq\{\varphi_{E}^{A}|A\in\mathrm{set\,of\,indexes},\varphi_{E}^{A}:E\rightarrow\mathbb{R}\}$
is introduced such that the index $A$ becomes $A=\mu$ if $E=\mathbb{V}_{\mathrm{M}}^{\mathrm{4}}$
and $A=\mu\nu$ when $E=\mathbb{V}_{\mathrm{M}}^{\mathrm{4}}\otimes\mathbb{V}_{\mathrm{M}}^{\mathrm{4}}$
($\mu,\nu=0,1,2,3$), even if we do not declare it explicitly.

\subsection{LAD model and QED correction}

In the non-relativistic classical dynamics, an electron emits its
energy by light by Larmor's formula such that $dW_{\mathrm{classical}}/dt=m_{0}\tau_{0}\dot{\boldsymbol{v}}^{2}$
\cite{Larmor}. Hence, the momentum transfer of an electron due to
its radiation has to be taken into account. It is associated by the
following energy balance relation:
\[
\frac{d}{dt}\left(\frac{1}{2}m_{0}\boldsymbol{v}^{2}\right)=\boldsymbol{F}_{\mathrm{ex}}\cdot\boldsymbol{v}-\frac{dW_{\mathrm{classical}}}{dt}
\]
Smearing this equation by its definite integral on the domain of $(-\infty,\infty)$
with the condition $[m_{0}\tau_{0}\dot{\boldsymbol{v}}\cdot\boldsymbol{v}]_{-\infty}^{\infty}=0$,
the Lorentz-Abraham (LA) equation is imposed:
\begin{equation}
m_{0}\dot{\boldsymbol{v}}=\boldsymbol{F}_{\mathrm{ex}}+m_{0}\tau_{0}\ddot{\boldsymbol{v}}\label{eq: LA}
\end{equation}
Then, the Lorentz-Abraham-Dirac (LAD) equation is derived as the covariant
form of (\ref{eq: LA}) in the Minkowski spacetime. The following
is the summary of the LAD model provided by P. A. M. Dirac \cite{Dirac(1938a)}.\begin{leftbar}
\begin{thm}[Classical dynamics]
\label{Classical_dynamics}Consider a single scalar electron system
denoted by the equation of motion and the Maxwell equation in $(\mathbb{A}^{4}(\mathbb{V_{\mathrm{M}}^{\mathrm{4}}},g),\mathscr{B}(\mathbb{A}^{4}(\mathbb{V_{\mathrm{M}}^{\mathrm{4}}},g)),\mu)$:
\begin{equation}
m_{0}\frac{dv^{\mu}}{d\tau}(\tau)=-e\left[F_{\mathrm{ex}}^{\mu\nu}(x(\tau))+F_{\mathrm{LAD}}^{\mu\nu}(x(\tau))\right]v_{\nu}(\tau)\label{eq:LAD-origin}
\end{equation}
\begin{equation}
\partial_{\mu}F^{\mu\nu}(x)=\mu_{0}\times\left[-ec\int_{\mathbb{R}}d\tau\,v^{\nu}(x(\tau))\delta^{4}(x-x(\tau))\right]\label{eq: classical Maxwell}
\end{equation}
Where $m_{0}$, $e$, $\mu_{0}$ and $c$ are the physical constants
of the rest mass, the charge of an electron, the vacuum permeability
and the speed of light. The vacuum permeability $\mu_{0}$ relates
to the vacuum permittivity $\varepsilon_{0}$, i.e., $\varepsilon_{0}\mu_{0}c^{2}=1$
in the SI unit. $v\in\mathbb{V_{\mathrm{M}}^{\mathrm{4}}}$ represents
the 4-velocity of an electron. The fields $F_{\mathrm{ex}}\in\mathbb{V_{\mathrm{M}}^{\mathrm{4}}}\otimes\mathbb{V_{\mathrm{M}}^{\mathrm{4}}}$
such that $\partial_{\mu}F_{\mathrm{ex}}^{\mu\nu}=0$ denotes an external
field and $F_{\mathrm{LAD}}\in\mathbb{V_{\mathrm{M}}^{\mathrm{4}}}\otimes\mathbb{V_{\mathrm{M}}^{\mathrm{4}}}$
satisfying $\partial_{\mu}F_{\mathrm{LAD}}^{\mu\nu}=0$ is, so-called,
''the LAD field (the effective radiation reaction field)'' which is
the homogeneous solution of (\ref{eq: classical Maxwell}).
\end{thm}
\end{leftbar}The homogeneous solution $F_{\mathrm{LAD}}$ is expressed
by the retarded field and the advanced field which are the solutions
of (\ref{eq: classical Maxwell}) \cite{Dirac(1938a),Hartemann(2001)}.\begin{leftbar}
\begin{lem}[Retarded, advanced and LAD fields]
\label{lemma_F_LAD}By solving (\ref{eq: classical Maxwell}), the
retarded field $F_{\mathrm{ret}}$ and the advanced field $F_{\mathrm{adv}}$
are obtained. By the substitution $x=x(\tau)$,
\begin{eqnarray}
F_{\epsilon}(x(\tau))|_{\epsilon=\mathrm{ret,adv}} & = & \frac{e}{8\pi\varepsilon_{0}c^{5}}\left[\frac{dv}{d\tau}(\tau)\otimes v(\tau)-v(\tau)\otimes\frac{dv}{d\tau}(\tau)\right]\int_{\mathbb{R}}d\tau'\frac{\delta(\tau'-\tau)}{|\tau'-\tau|}\nonumber \\
 &  & -\mathrm{sign}(\epsilon)\times\frac{e}{6\pi\varepsilon_{0}c^{5}}\left[\frac{d^{2}v}{d\tau^{2}}(\tau)\otimes v(\tau)-v(\tau)\otimes\frac{d^{2}v}{d\tau^{2}}(\tau)\right]\,.
\end{eqnarray}
Where, $\mathrm{sign}(\mathrm{ret})=1$ and $\mathrm{sign}(\mathrm{adv})=-1$.
The LAD field $F_{\mathrm{LAD}}$ is given by 
\begin{eqnarray}
F_{\mathrm{LAD}}(x(\tau)) & \coloneqq & \frac{F_{\mathrm{ret}}(x(\tau))-F_{\mathrm{adv}}(x(\tau))}{2}\nonumber \\
 & = & -\frac{e}{6\pi\varepsilon_{0}c^{5}}\left[\frac{d^{2}v}{d\tau^{2}}(\tau)\otimes v(\tau)-v(\tau)\otimes\frac{d^{2}v}{d\tau^{2}}(\tau)\right]\,,\label{eq: F_LAD}
\end{eqnarray}
as the homogeneous solution of the Maxwell equation (\ref{eq: classical Maxwell})
at $x=x(\tau)$. 
\end{lem}
\end{leftbar}\begin{snugshade}\begin{leftbar}
\begin{defn}[LAD equation]
\label{LAD}By substituting (\ref{eq: F_LAD}) for (\ref{eq:LAD-origin}),
\begin{equation}
m_{0}\frac{dv^{\mu}}{d\tau}(\tau)=-eF_{\mathrm{ex}}^{\mu\nu}(x(\tau))v_{\nu}(\tau)+\frac{m_{0}\tau_{0}}{c^{2}}\left[\frac{d^{2}v^{\mu}}{d\tau^{2}}(\tau)\cdot v^{\nu}(\tau)-\frac{d^{2}v^{\nu}}{d\tau^{2}}(\tau)\cdot v^{\mu}(\tau)\right]v_{\nu}(\tau)\label{eq:LAD}
\end{equation}
is named the Lorentz-Abraham-Dirac (LAD) equation \cite{Dirac(1938a)}.
Where, $\tau_{0}\coloneqq e^{2}/6\pi\varepsilon_{0}m_{0}c^{3}=O(10^{-24}\mathrm{sec})$.
The non-relativistic limit of (\ref{eq:LAD-origin}) tends to (\ref{eq: LA}).
The energy loss by radiation is estimated by the invariant form of
Larmor's formula:
\begin{eqnarray}
\frac{dW_{\mathrm{classical}}}{dt}(\tau) & = & -m_{0}\tau_{0}\frac{dv_{\alpha}}{d\tau}(\tau)\cdot\frac{dv^{\alpha}}{d\tau}(\tau)\;(\geq0)\label{eq: dW_classical/dt}\\
 & = & -\frac{\tau_{0}}{m_{0}}g_{\alpha\beta}\left(-eF_{\mathrm{ex}}^{\alpha\mu}v_{\mu}\right)\left(-eF_{\mathrm{ex}}^{\beta\nu}v_{\nu}\right)+O(\tau_{0}^{2})
\end{eqnarray}

\end{defn}
\end{leftbar}\end{snugshade}

This LAD equation has been treated as the standard model of radiation
reaction in the previous laser-plasma simulations for explaining radiation
from classical particles since $d(m_{o}c^{2}\gamma)/dt\approx\boldsymbol{F}_{\mathrm{ex}}\cdot\boldsymbol{v}-dW_{\mathrm{classical}}/dt$
{[}the 0th component of (\ref{eq:LAD}){]}. However, its quantum correction
has started to be considered in the recent works provided by high-intensity
lasers \cite{Seto(2015a),I. Sokolov(2011a)}. It is found by the following
steps: For an external field of {\bf a plane wave}, Volkov solutions
of a Dirac equation are derived \cite{Volkov(1935),Bergou(1980)}.
Then, a family of these Volkov solutions imposes its orthogonality
and completeness \cite{Zakowicz(2005),Boca-Florescu(2010)}, thus,
the family is regarded as a basis of Dirac fields absorbing an external
field (a dressed Dirac field) \cite{Bergou(1980),Berestetskii-Lifshitz QED}.
An $S$-matrix by a family of Volkov solutions is introduced \cite{Furry(1951),Ritus(1972a)}.
Finally, a non-linear Compton scattering as radiation reaction can
be calculated \cite{Brown-Kibble(1964),Nikishov(1964a)),Nikishov(1964b)}.
This is so-called the Furry picture. By resulting them, the radiation
formula is derived as follows \cite{Seto(2015a),I. Sokolov(2011a),A. Sokolov(1986)}:
\begin{equation}
\boxed{\frac{dW_{\mathrm{QED}}}{dt}=q(\chi)\times\frac{dW_{\mathrm{classical}}}{dt}}\label{eq: Q=000026C Larmor}
\end{equation}
Thus, the two physical regimes between classical dynamics and QED
are bridged by $q(\chi)$. Therefore, $q(\chi)$ is regarded as a
quantumness of radiation. Where, 
\begin{equation}
\chi\coloneqq\frac{3}{2}\frac{\hbar}{m_{0}^{2}c^{3}}\sqrt{-g_{\mu\nu}(-eF_{\mathrm{ex}}^{\mu\alpha}v_{\alpha})(-eF_{\mathrm{ex}}^{\nu\beta}v_{\beta})}\,,
\end{equation}
\begin{equation}
q(\chi)=\frac{9\sqrt{3}}{8\pi}\int_{0}^{\chi^{-1}}dr\,r\left[\int_{\frac{r}{1-\chi r}}^{\infty}dr'K_{5/3}(r')+\frac{\chi^{2}r^{2}}{1-\chi r}K_{2/3}\left(\frac{r}{1-\chi r}\right)\right]\,.\label{eq: q_chi}
\end{equation}
Many authors have investigated this effect for their future experiments
by high-intensity lasers \cite{Seto(2015a),I. Sokolov(2011a),Koga(2005),Harvey(2009),Mackenroth(2011),Neitz(2013),Seipt(2013),Anton(2013),RA5,Roos(1966),Fried(1966)}.
For example at ELI-NP in Romania \cite{ELI-NP,RA5}, they aim the
factor $q(\chi)=0.3$ by a combination between a high-intensity laser
of $O(10^{22}\mathrm{W/cm^{2}})$ and an electron of $O(1\mathrm{GeV})$
{[}see Figure \ref{q_chi}{]} \cite{Seto(2015a)}. Thus, the behavior
of $q(\chi)$ {\underline {\bf depending on laser intensities}} is
important to understand the transition from classical dynamics to
non-linear QED. 
\begin{figure}
\noindent \centering{}\includegraphics[scale=1.3]{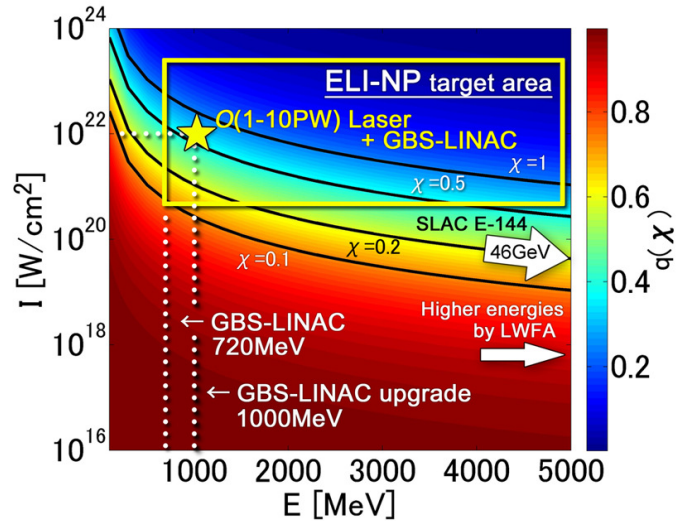}\caption{\label{q_chi}Quantumness of radiation, $q(\chi)$. This is the plot
of (\ref{eq: q_chi}) with respect to energies of an electron and
laser intensities. When we choose the combination between a laser
intensity of $O(10^{22}\mathrm{W}/\mathrm{cm}^{2})$ and an electron
energy of $O(1\mathrm{GeV})$, the factor $q(\chi)$ reaches $0.3$
which is the feasible regime produced by the ELI-NP facility \cite{RA5}.
The SLAC E-144 included the experiments of the non-linear Compton
scatterings by the combination of\textcolor{red}{{} }$O(10^{18}\mathrm{W}/\mathrm{cm}^{2}+46\mathrm{GeV})$\textcolor{red}{{}
\cite{SLAC}}. }
\end{figure}

However, can we generalize the applicable range of this Furry picture
from an external field of {\bf a plane wave}? Unfortunately, there
are the mathematical demonstrations only in the case of a plane wave
\cite{Zakowicz(2005),Boca-Florescu(2010),Guo(1991)}. It is a major
interest for the expression of strong focused or superpositioned lasers
{[}A. Di Piazza{]}. In addition, quantum dynamics doesn't provide
{\bf the method to draw a real trajectory of a particle}. It is very
helpful for laser-plasma simulations if it is improved. Let us aim
the second purpose in this article.

\subsection{Proposal: by a Brownian particle and fields}

\begin{figure}
\centering{}\includegraphics[scale=0.65]{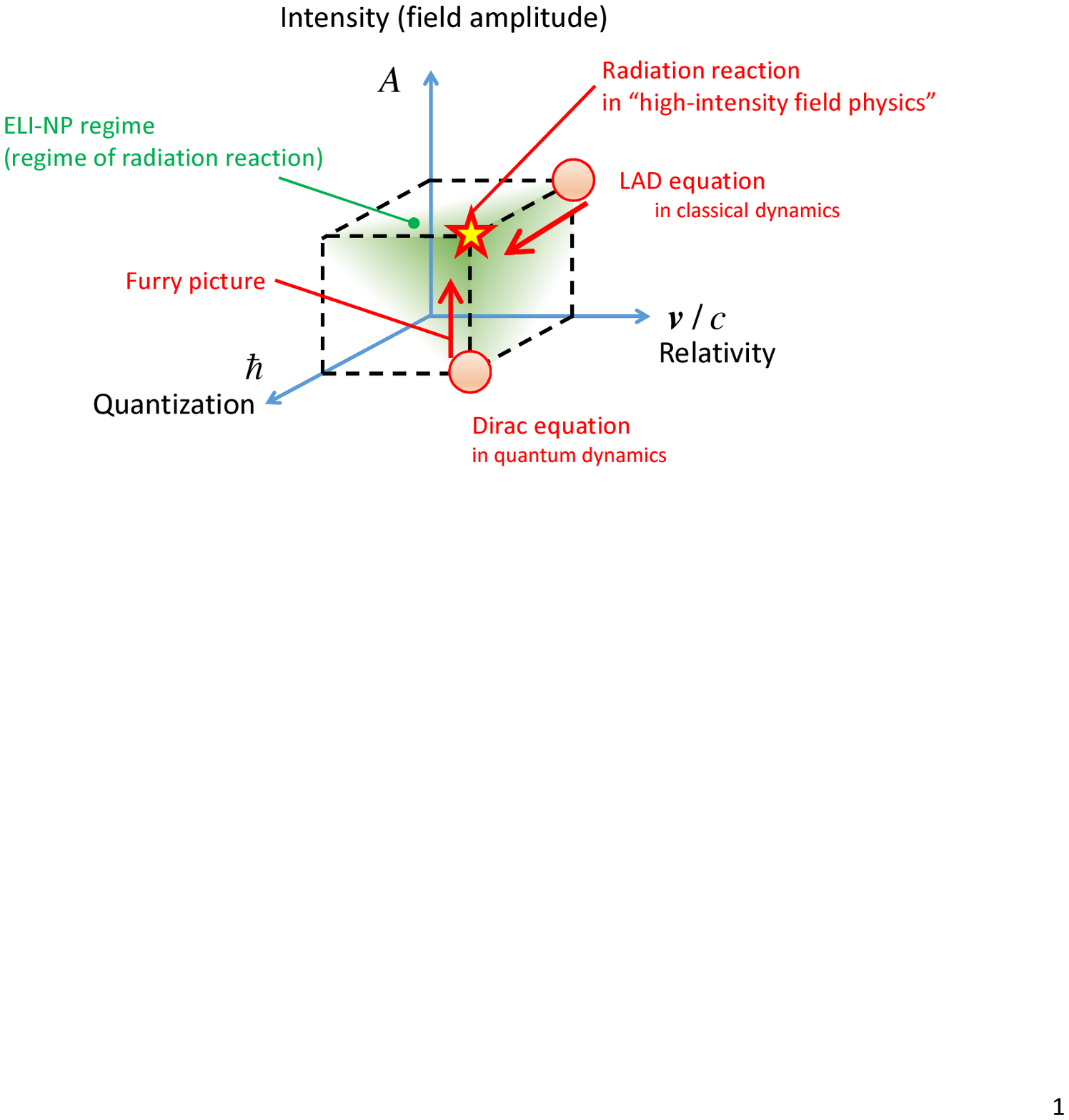}\caption{The physical regime of an electron. The point at ''the star'' is the
regime of high-intensity field physics. The Furry picture is the first
way from relativistic quantum dynamics to the point of ''the star''.
On the other hand, the quantization after reaching high-intensity
``classical'' dynamics is its second candidate. The Lorentz-Abraham-Dirac
(LAD) equation at high-intensity ``classical'' dynamics ({\bf Definition \ref{LAD}})
is the standard model of radiation reaction. ELI-NP is the state-of-the-arts
laser facility which can perform the green area and proposes the real
experiments of radiation reaction \cite{ELI-NP,RA5}\label{figHIFP}\setstretch{1.2}}
\end{figure}

The regime of {\bf high-intensity field physics} can be found at
quantized, relativistic and high-intensity field interactions marked
by ''the star'' in {\bf Figure \ref{figHIFP}}. The Furry picture
is the way from the relativistic quantum dynamics. The quantization
from the LAD equation at classical, relativistic and high-intensity
regime shall be another candidate. Therefore, we discuss its feasibility
by Nelson's stochastic quantization \cite{Nelson(1966a),Nelson(2001_book)},
holding the highly compatible expressions between classical and quantum
dynamics. The basic construction of relativistic scalar quantum dynamics
with its Brownian kinematics has been given in {\bf Volume I} \cite{Seto Volume I},
let us summarize those ideas. \begin{leftbar}
\begin{defn}[$\{\mathscr{P}_{\tau}\}$ and $\{\mathscr{F}_{\tau}\}$ \cite{Seto Volume I}]
For a probability space $\left(\mathit{\Omega}^{1\mathchar`-\mathrm{dim}},\mathcal{F}^{1\mathchar`-\mathrm{dim}},\mathscr{P}^{1\mathchar`-\mathrm{dim}}\right)$,
let $\{\mathcal{P}_{\tau}\}_{\tau\in\mathbb{R}}$ and $\{\mathcal{F}_{\tau}\}_{\tau\in\mathbb{R}}$
are increasing and decreasing families of 1-dimensional continuous
processes such that $\mathcal{P}_{\tau}\subset\mathcal{F}^{1\mathchar`-\mathrm{dim}}$
satisfies $-\infty<\sigma\leq\tau\Longrightarrow\mathcal{P}_{\sigma}\subset\mathcal{P}_{\tau}$
and $\mathcal{F}_{\tau}\subset\mathcal{F}^{1\mathchar`-\mathrm{dim}}$
fulfills $\tau\leq\sigma<\infty\Longrightarrow\mathcal{F}_{\tau}\supset\mathcal{F}_{\sigma}$.
Then for $\varOmega\coloneqq\times^{4}\mathit{\Omega}^{1\mathchar`-\mathrm{dim}}$,
$\mathcal{F}\coloneqq\otimes^{4}\mathcal{F}^{1\mathchar`-\mathrm{dim}}$
and $\mathscr{P}\coloneqq\times^{4}\mathscr{P}^{1\mathchar`-\mathrm{dim}}$,
there are 4-dimensional stochastic processes on $(\varOmega,\mathcal{F},\mathscr{P})$
characterized by the following; $\{\mathscr{P}_{\tau}\}_{\tau\in\mathbb{R}}$
is the family of the composited sub-$\sigma$-algebra $\mathscr{P}_{\tau\in\mathbb{R}}\coloneqq\mathcal{F}_{\tau}\otimes\mathcal{P}_{\tau}\otimes\mathcal{P}_{\tau}\otimes\mathcal{P}_{\tau}\in\mathcal{F}$
and $\{\mathscr{F}_{\tau}\}_{\tau\in\mathbb{R}}$ is the family of
$\mathscr{F}_{\tau\in\mathbb{R}}\coloneqq\mathcal{P}_{\tau}\otimes\mathcal{F}_{\tau}\otimes\mathcal{F}_{\tau}\otimes\mathcal{F}_{\tau}\in\mathcal{F}$.
\end{defn}
\end{leftbar}Consider two measurable spaces $(X,\mathcal{X})$ and
$(Y,\mathcal{Y})$, a $\mathcal{X}/\mathcal{Y}$-measurable mapping
$f$ is a mapping $f:X\rightarrow Y$ satisfying $f^{-1}(A)\coloneqq\{x\in X|f(x)\in A\}\subset\mathcal{X}$
for all $A\in\mathcal{Y}$. By employing $\mathbb{L}_{T}^{p}(E)$
 as a family of $\mathscr{B}(\mathbb{R})\times\mathcal{F}/\mathscr{B}(E)$-measurable
mappings for a topological space $E$, let us introduce $\mathcal{L}_{\mathrm{loc}}^{p}(\{\mathscr{P}_{\tau}\};E)$
and $\mathcal{L}_{\mathrm{loc}}^{p}(\{\mathscr{F}_{\tau}\};E)$ as
the families of stochastic processes as follows:

\[
\mathbb{L}_{T}^{p}(E)\coloneqq\left\{ \hat{X}(\circ,\bullet)\left|\hat{X}(\circ,\bullet):\mathbb{R}\times\varOmega\rightarrow E,\sum_{A}\int_{T\subset\mathbb{R}}|\varphi_{E}^{A}\circ\hat{X}(\tau,\omega)|^{p}d\tau<\infty\,\mathrm{a.s.}\right.\right\} 
\]
\[
\mathcal{L}_{\mathrm{loc}}^{p}(\{\mathscr{P}_{\tau}\};E)\coloneqq\left\{ \left.\hat{X}(\circ,\bullet)\in\mathbb{L}_{(-\infty,\tau]}^{p}(E)\right|\begin{gathered}\hat{X}(\circ,\bullet)\mathrm{\,is\,}\{\mathscr{P}_{\tau}\}\mathchar`-\mathrm{adapted}\end{gathered}
\right\} 
\]
\[
\mathcal{L}_{\mathrm{loc}}^{p}(\{\mathscr{F}_{\tau}\};E)\coloneqq\left\{ \left.\hat{X}(\circ,\bullet)\in\mathbb{L}_{[\tau,\infty)}^{p}(E)\right|\forall\tau\in\mathbb{R},\,\hat{X}^{\mu}(\circ,\bullet)\mathrm{\,is\,}\{\mathscr{F}_{\tau}\}\mathchar`-\mathrm{adapted}\right\} 
\]
\begin{leftbar}
\begin{defn}[$(-g)$-Wiener processes of $W_{+}(\circ,\bullet)$ and $W_{-}(\circ,\bullet)$
\cite{Seto Volume I}]
Let $W_{+}(\circ,\bullet)$ and $W_{-}(\circ,\bullet)$ be $\{\mathscr{P}_{\tau}\}$
and $\{\mathscr{F}_{\tau}\}$-$(-g)$ Wiener processes satisfying
the following for $\mu,\nu=0,1,2,3$:
\begin{equation}
\begin{array}{cc}
\mathbb{E}\left\llbracket \int_{\tau}^{\tau+\delta\tau}dW_{\pm}^{\mu}(\tau',\bullet)\right\rrbracket =0\,, & \mathbb{E}\left\llbracket \int_{\tau}^{\tau+\delta\tau}dW_{\pm}^{\mu}(\tau',\bullet)\times\int_{\tau}^{\tau+\delta\tau}dW_{\pm}^{\mu}(\tau'',\bullet)\right\rrbracket =\delta^{\mu\nu}\times\delta\tau\,,\end{array}
\end{equation}
The It\^{o} formula of a $C^{2}$-function $f:\mathbb{A}^{4}(\mathbb{V_{\mathrm{M}}^{\mathrm{4}}},g)\rightarrow\mathbb{C}$
on $W_{\pm}(\circ,\omega)$ is
\begin{equation}
f(W_{\pm}(\tau_{b},\omega))-f(W_{\pm}(\tau_{a},\omega))=\int_{\tau_{a}}^{\tau_{b}}\partial_{\mu}f(W_{\pm}(\tau,\omega))dW_{\pm}^{\mu}(\tau,\omega)\mp\frac{\lambda^{2}}{2}\int_{\tau_{a}}^{\tau_{b}}\partial_{\mu}\partial^{\mu}f(W_{\pm}(\tau,\omega))d\tau\,\,\mathrm{a.s.}
\end{equation}

\end{defn}
\end{leftbar}The following stochastic differential equation is defined
as a relativistic kinematics of a stochastic scalar electron by employing
$W_{\pm}(\circ,\bullet)$:\begin{leftbar}
\begin{thm}[Kinematics \cite{Seto Volume I}]
For $\left(\mathit{\Omega},\mathcal{F},\mathscr{P}\right)$, ``A
dual-progressively measurable process'', or by shortening ``A D-process''
and ``A D-progressive'' $\hat{x}(\circ,\bullet)\coloneqq\{\hat{x}(\tau,\omega)\in\mathbb{A}^{4}(\mathbb{V_{\mathrm{M}}^{\mathrm{4}}},g)|\tau\in\mathbb{R},\omega\in\varOmega\}$
is a 4-dimensional $\{\mathscr{P}_{\tau}\cap\mathscr{F}_{\tau}\}$-adapted
process, such that
\begin{equation}
\ensuremath{d_{\pm}\hat{x}(\tau,\omega)=\mathcal{V}_{\pm}(\hat{x}(\tau,\omega))d\tau+\lambda\times dW_{\pm}(\tau,\omega)}\,,\label{eq: kinematics}
\end{equation}
where, $\lambda=\sqrt{\hbar/m_{0}}$, $\mathcal{V}_{+}(\hat{x}(\circ,\bullet))\in\mathcal{L}_{\mathrm{loc}}^{1}(\{\mathscr{P}_{\tau}\};\mathbb{V}_{\mathrm{M}}^{\mathrm{4}})$
and $\mathcal{V}_{-}(\hat{x}(\circ,\bullet))\in\mathcal{L}_{\mathrm{loc}}^{1}(\{\mathscr{F}_{\tau}\};\mathbb{V}_{\mathrm{M}}^{\mathrm{4}})$.
For a $C^{2}$-function $f:\mathbb{A}^{4}(\mathbb{V_{\mathrm{M}}^{\mathrm{4}}},g)\rightarrow\mathbb{C}$,
its It\^{o} formula $d_{\pm}f$ is characterized by 
\begin{equation}
d_{\pm}f(\hat{x}(\tau,\omega))=\partial_{\mu}f(\hat{x}(\tau,\omega))d_{\pm}\hat{x}^{\mu}(\tau,\omega)\mp\frac{\lambda^{2}}{2}\partial_{\mu}\partial^{\mu}f(\hat{x}(\tau,\omega))d\tau\,\,\mathrm{a.s.}\label{eq:2-Ito-formula}
\end{equation}
 The probability density of $\hat{x}(\circ,\bullet)$, i.e., $p(x,\tau)$
satisfies the following Fokker-Planck equation:
\begin{equation}
\partial_{\tau}p(x,\tau)+\partial_{\mu}[\mathcal{V}_{\pm}^{\mu}(x)p(x,\tau)]\pm\frac{\lambda^{2}}{2}\partial^{\mu}\partial_{\mu}p(x,\tau)=0\label{eq: Fokker-Planck}
\end{equation}

\end{thm}
\end{leftbar}For $\left(\mathit{\Omega},\mathcal{F},\mathscr{P}\right)$,
let $\mathbb{E}\llbracket\hat{X}(\bullet)\rrbracket$ be the expectation
of the random variable $\hat{X}(\bullet)\coloneqq\{\hat{X}(\omega)|\omega\in\varOmega\}$,
namely, $\mathbb{E}\llbracket\hat{X}(\bullet)\rrbracket\coloneqq\int_{\Omega}d\mathscr{P}(\omega)\,\hat{X}(\omega)$.
The conditional probability of $B$ given $A$ is denoted by $\mathscr{P}_{A}(B)$.
For $\{A_{n}\}_{n=1}^{\infty}$ of a countable decomposition of $\varOmega$,
its minimum $\sigma$-algebra $\mathcal{\mathscr{C}}=\sigma(\{A_{n}\}_{n=1}^{\infty})$
is introduced.  $\mathbb{E}\llbracket\hat{X}(\bullet)|\mathcal{\mathscr{C}}\rrbracket(\omega)\coloneqq\sum_{n=1}^{\infty}\{\int_{\varOmega}d\mathscr{P}_{A_{n}}(\omega')\hat{X}(\omega')\}\mathbf{1}_{A_{n}}(\omega)$
is defined as the conditional expectation of $\hat{X}(\bullet)$ given
$\mathcal{\mathscr{C}}\subset\mathcal{F}$; $\mathbf{1}_{A_{n}}(\omega)$
satisfies $\mathbf{1}_{A_{n}}(\omega\in A_{n})=1$ and $\mathbf{1}_{A_{n}}(\omega\notin A_{n})=0$.
Since a D-progressive $\hat{x}(\tau,\omega)$ is $\{\mathscr{P}_{\tau}\cap\mathscr{F}_{\tau}\}$-adapted,
the following complex differential and velocity can be defined;
\begin{defn}[Complex differential and velocity]
\begin{leftbar}Let $\hat{d}\coloneqq(1-i)/2\times d_{+}+(1+i)/2\times d_{-}$
be the complex differential on a given D-progressive $\hat{x}(\circ,\bullet)$
for $f:\mathbb{A}^{4}(\mathbb{V_{\mathrm{M}}^{\mathrm{4}}},g)\rightarrow\mathbb{C}$:
\begin{equation}
\hat{d}f(\hat{x}(\tau,\omega))=\partial_{\mu}f(\hat{x}(\tau,\omega))\hat{d}\hat{x}^{\mu}(\tau,\omega)+\frac{i\lambda^{2}}{2}\partial^{\mu}\partial_{\mu}f(\hat{x}(\tau,\omega))d\tau\,\mathrm{a.s.}
\end{equation}
 Then, consider the conditional expectation of the derivative given
$\gamma_{\tau}\coloneqq\mathscr{P}_{\tau}\cap\mathscr{F}_{\tau}\subset\mathcal{F}$
is denoted by
\begin{equation}
\mathbb{E}\left\llbracket \left.\frac{\hat{d}f}{d\tau}(\hat{x}^{\mu}(\tau,\bullet))\right|\gamma_{\tau}\right\rrbracket (\omega)=\mathscr{\mathcal{V}}^{\mu}(\hat{x}(\tau,\omega))\partial_{\mu}f(\hat{x}(\tau,\omega))+\frac{i\lambda^{2}}{2}\partial^{\mu}\partial_{\mu}f(\hat{x}(\tau,\omega))\,.
\end{equation}
Especially when $f(\hat{x}(\tau,\omega))=\hat{x}(\tau,\omega)$, it
derives the complex velocity $\mathscr{\mathcal{V}}\in\mathbb{V_{\mathrm{M}}^{\mathrm{4}}}\oplus i\mathbb{V_{\mathrm{M}}^{\mathrm{4}}}$,
\begin{equation}
\mathscr{\mathcal{V}}^{\mu}(\hat{x}(\tau,\omega))\coloneqq\mathbb{E}\left\llbracket \left.\frac{\hat{d}\hat{x}^{\mu}}{d\tau}(\tau,\bullet)\right|\gamma_{\tau}\right\rrbracket (\omega)=\frac{1-i}{2}\mathcal{V}_{+}^{\mu}(\hat{x}(\tau,\omega))+\frac{1+i}{2}\mathcal{V}_{-}^{\mu}(\hat{x}(\tau,\omega))\,.\label{eq: complex V}
\end{equation}
\end{leftbar}
\end{defn}
The following signatures are employed for the simple description:
\begin{equation}
\hat{\mathcal{V}}^{\mu}(x)\coloneqq\mathcal{V}^{\mu}(x)+\frac{i\lambda^{2}}{2}\partial^{\mu}\label{eq: OP of complex V}
\end{equation}
\begin{equation}
\mathfrak{D_{\tau}}=\hat{\mathcal{V}}^{\mu}(x)\partial_{\mu}\label{eq: complex D_tau}
\end{equation}
 \begin{leftbar}
\begin{thm}[Dynamics \cite{Seto Volume I}]
\label{Quantum_Dynamics}The dynamics of a stochastic scalar electron
and a field coupled with a kinematics of (\ref{eq: kinematics}) is
as follows:\begin{snugshade}
\begin{equation}
m_{0}\mathfrak{D_{\tau}}\mathcal{V}^{\mu}(\hat{x}(\tau,\omega))=-e\mathcal{\hat{V}}_{\nu}(\hat{x}(\tau,\omega))F^{\mu\nu}(\hat{x}(\tau,\omega))\label{eq:sug-KG eq}
\end{equation}
\begin{equation}
\partial_{\mu}\left[F^{\mu\nu}(x)+\delta f^{\mu\nu}(x)\right]=\mu_{0}\times\mathbb{E}\left\llbracket -ec\int_{\mathbb{R}}d\tau'\,\mathrm{Re}\left\{ \mathcal{V}^{\nu}(x)\right\} \delta^{4}(x-\hat{x}(\tau',\bullet))\right\rrbracket \label{eq:sug-Maxwell}
\end{equation}
\end{snugshade}

\noindent Where, (\ref{eq:sug-KG eq}) is equivalent to the Klein-Gordon
equation via $\mathcal{V}^{\mu}(x)\coloneqq i\lambda^{2}\times\partial{}^{\mu}\ln\phi(x)+e/m_{0}\times A{}^{\mu}(x)$
and (\ref{eq:sug-Maxwell}) is the Maxwell equation. Let $F\in\mathbb{V_{\mathrm{M}}^{\mathrm{4}}}\otimes\mathbb{V_{\mathrm{M}}^{\mathrm{4}}}$
be a homogeneous solution of (\ref{eq:sug-Maxwell}) and $\delta f\in\mathbb{V_{\mathrm{M}}^{\mathrm{4}}}\otimes\mathbb{V_{\mathrm{M}}^{\mathrm{4}}}$
be a singularity of the field. The dynamics of (\ref{eq:sug-KG eq})
and (\ref{eq:sug-Maxwell}) satisfy the $U(1)$-gauge symmetry. 
\end{thm}
\noindent \end{leftbar}

By using those ideas, a new parameter is found in Larmor's formula
instead of $q(\chi)$ of (\ref{eq: Q=000026C Larmor}-\ref{eq: q_chi}).
\\

In {\bf Section \ref{Radiation}}, the retarded field $\mathcal{F}_{(\mathring{+})}$
and the advanced field $\mathcal{F}_{(\mathring{-})}$ are derived
from the Maxwell equation (\ref{eq:sug-Maxwell}) by its direct calculation.
Then, the effective radiation reaction field $\mathfrak{F}(\hat{x}(\tau,\omega))$
corresponding to $F_{\mathrm{LAD}}$ of (\ref{eq: F_LAD}) is defined.
Where, the key issue is the treatment of the current density in (\ref{eq:sug-Maxwell});
\[
j_{\mathrm{stochastic}}(x)=\mathbb{E}\left\llbracket -ec\int_{\mathbb{R}}d\tau'\,\mathrm{Re}\left\{ \mathcal{V}(x)\right\} \delta^{4}(x-\hat{x}(\tau',\bullet))\right\rrbracket \in\mathbb{V}_{\mathrm{M}}^{4}\,.
\]

Then, the characteristics of the present model is discussed in {\bf Section \ref{Sect-Summary}}
as the summary. Ehrenfest's theorem is applied to explain the transition
between quantum dynamics to classical dynamics and the Landau-Lifshitz
approximation for the real calculation. By resulting it, we can find
a simple radiation formula corresponding to (\ref{eq: Q=000026C Larmor}):\begin{snugshade}
\[
\frac{dW}{dt}(\mathbb{\mathbb{E}}\llbracket\hat{x}(\tau,\bullet)\rrbracket)=\int_{\omega\in\varOmega_{\tau}^{\mathrm{ave}}}d\mathscr{P}(\omega)\,\frac{dW_{\mathrm{classical}}}{dt}(\hat{x}(\tau,\omega))
\]
\end{snugshade}

\section{Radiation fields\label{Radiation}}

Let us consider how to formalize a radiation field by solving the
Maxwell equation (\ref{eq:sug-Maxwell}). By recalling {\bf Theorem \ref{Classical_dynamics}}
and {\bf Lemma \ref{lemma_F_LAD}}, the retarded field $F_{\mathrm{ret}}$
and advanced field $F_{\mathrm{adv}}$ are derived from the Maxwell
equation (\ref{eq: classical Maxwell}) in classical physics. $F_{\mathrm{ret}}$
is separated into its homogeneous solution and its singularity at
$x=x(\tau)$:
\begin{equation}
\frac{F_{\mathrm{ret}}(x(\tau))-F_{\mathrm{adv}}(x(\tau))}{2}=-\frac{m_{0}\tau_{0}}{ec^{2}}\left[\frac{d^{2}v}{d\tau^{2}}(\tau)\otimes v(\tau)-v(\tau)\otimes\frac{d^{2}v}{d\tau^{2}}(\tau)\right]\label{eq: classical-hom}
\end{equation}
 
\begin{equation}
\frac{F_{\mathrm{ret}}(x(\tau))+F_{\mathrm{adv}}(x(\tau))}{2}=\frac{3}{4}\frac{m_{0}\tau_{0}}{ec^{2}}\left[\frac{dv}{d\tau}(\tau)\otimes v(\tau)-v(\tau)\otimes\frac{dv}{d\tau}(\tau)\right]\int_{\mathbb{R}}d\tau'\frac{\delta(\tau'-\tau)}{|\tau'-\tau|}\label{eq: classical-singular}
\end{equation}
The singularity of $\delta f\coloneqq(F_{\mathrm{ret}}+F_{\mathrm{adv}})/2$
is known as the self-interaction of an electron due to its electromagnetic
mass $m_{\mathrm{EM}}\coloneqq3/4\times m_{0}\tau_{0}\int_{\mathbb{R}}d\tau'\delta(\tau'-\tau)/|\tau'-\tau|=\infty$
not to depend on an external field $F_{\mathrm{ex}}$ \cite{Dirac(1938a),Hartemann(2001)}.
Thus, let us regard $\delta f$ as a Coulomb field attached to an
electron. The field $F_{\mathrm{LAD}}\coloneqq(F_{\mathrm{ret}}-F_{\mathrm{adv}})/2$
of its rest denotes the LAD field (the effective field), such that
$\partial_{\mu}F_{\mathrm{LAD}}^{\mu\nu}=0$ for all $x\in\mathbb{A}^{4}(\mathbb{V_{\mathrm{M}}^{\mathrm{4}}},g)$.
Consider the one in stochastic dynamics of a scalar electron. {\bf Theorem \ref{Quantum_Dynamics}}
and following {\bf Lemma \ref{F_LAD_stochastic}}'' are ones in the
present model corresponding to ''{\bf Theorem \ref{Classical_dynamics}}
and {\bf Lemma \ref{lemma_F_LAD}}''.

\subsection{Derivation of fields: general idea}

\noindent \begin{leftbar}
\begin{lem}[Radiation fields]
\noindent \label{F_LAD_stochastic}Consider the Maxwell equation
\begin{equation}
\partial_{\mu}\mathcal{F}^{\mu\nu}(x)=\mu_{0}\times\mathbb{E}\left\llbracket -ec\int_{\mathbb{R}}d\tau\,\mathrm{Re}\left\{ \mathcal{V}^{\nu}(x)\right\} \delta^{4}(x-\hat{x}(\tau,\bullet))\right\rrbracket \,,\label{eq: Maxwell_for_derivation}
\end{equation}
let $\mathcal{F}_{(\mathring{+})}\in\mathbb{V}_{\mathrm{M}}^{\mathrm{4}}\otimes\mathbb{V}_{\mathrm{M}}^{\mathrm{4}}$
and $\mathcal{F}_{(\mathring{-})}\in\mathbb{V}_{\mathrm{M}}^{\mathrm{4}}\otimes\mathbb{V}_{\mathrm{M}}^{\mathrm{4}}$
be the retarded and advanced fields as the linear independent solutions
of (\ref{eq: Maxwell_for_derivation}). The particular solution of
(\ref{eq: Maxwell_for_derivation}) is derived as follows:
\begin{eqnarray}
\mathcal{F}(x) & = & F_{\mathrm{ex}}(x)+\mathcal{F}_{(\mathring{+})}(x)\nonumber \\
 & = & \mathring{+}\left[F_{\mathrm{ex}}(x)+\mathfrak{F}(x)\right]+\delta\mathfrak{f}(x)\label{eq: relation of radiation field}
\end{eqnarray}
Where, $\mathcal{F}_{(\mathring{+})}$ is separated into the effective
radiation (reaction) field $\mathfrak{F}$ such that $\partial_{\mu}\mathfrak{F}^{\mu\nu}=0$
and its singularity $\delta\mathfrak{f}\in\mathbb{V}_{\mathrm{M}}^{\mathrm{4}}\otimes\mathbb{V}_{\mathrm{M}}^{\mathrm{4}}$. 
\end{lem}
\noindent \end{leftbar}

\noindent (\ref{eq: Maxwell_for_derivation}) is analyzed by the following
decomposition:
\begin{equation}
\partial_{\mu}\mathcal{F}_{(\mathring{\pm})}^{\mu\nu}(x)=\mu_{0}j_{\mathrm{stochastic}}^{\nu}(x)\label{eq: Maxwell equation for solving}
\end{equation}
\begin{equation}
\mathcal{F}_{(\mathring{\pm})}^{\mu\nu}(x)\coloneqq\partial^{\mu}\mathcal{A}_{(\mathring{\pm})}^{\nu}(x)-\partial^{\nu}\mathcal{A}_{(\mathring{\pm})}^{\mu}(x)\label{eq: def-F}
\end{equation}
\begin{equation}
j_{\mathrm{stochastic}}^{\nu}(x)\coloneqq-ec\int_{\mathbb{R}}d\tau'\mathbb{E}\left\llbracket \mathrm{Re}\left\{ \mathcal{V}^{\nu}(x)\right\} \delta^{4}(x-\hat{x}(\tau,\bullet))\right\rrbracket \label{eq: j-stochastic}
\end{equation}
By selecting the Lorenz gauge $\partial_{\mu}\mathcal{A}_{(\mathring{\pm})}^{\mu}(x)=0$,
\begin{equation}
\partial_{\mu}\partial^{\mu}\mathcal{A}_{(\mathring{\pm})}^{\nu}(x)=\mu_{0}j_{\mathrm{stochastic}}^{\nu}(x)\,.\label{eq: Maxwell for A}
\end{equation}
Let us introduce the Green function $G_{\mathrm{(\mathring{\pm})}}(x,x')=\theta(\mathring{\pm}\Delta x^{0})/2\pi\times\delta(\Delta x_{\mu}\Delta x^{\mu})$
satisfying $\partial_{\alpha}\partial^{\alpha}G_{\mathrm{(\mathring{\pm})}}(x,x')=\delta^{4}(x-x')$
with the help of the notation $\Delta x\coloneqq x-x'$ and the unit
step function $\theta$; $\theta(x\geq0)=1$ and $\theta(x<0)=0$.
Thus, the solution of (\ref{eq: Maxwell for A}) is
\begin{equation}
\mathcal{A}_{(\mathring{\pm})}^{\nu}(x)=-ec\mu_{0}\int_{\mathbb{R}}d\tau'\,\mathbb{E}\left\llbracket \mathrm{Re}\left\{ \mathcal{V}^{\nu}(\hat{x}(\tau',\bullet))\right\} G_{\mathrm{(\mathring{\pm})}}(x,\hat{x}(\tau',\bullet))\right\rrbracket 
\end{equation}

\noindent and its field strength $\mathcal{F}_{(\mathring{\pm})}^{\mu\nu}$
by (\ref{eq: def-F}) becomes 
\begin{equation}
\mathcal{F}_{(\mathring{\pm})}^{\mu\nu}(x)=-\frac{e}{c\varepsilon_{0}}\int_{\mathbb{R}}d\tau'\,\mathbb{E}\left\llbracket \left(\mathrm{Re}\left\{ \mathcal{V}^{\nu}(\hat{x}(\tau',\bullet))\right\} \partial^{\mu}-\mathrm{Re}\left\{ \mathcal{V}^{\nu}(\hat{x}(\tau',\bullet))\right\} \partial^{\mu}\right)G_{\mathrm{(\mathring{\pm})}}(x,\hat{x}(\tau',\bullet))\right\rrbracket \,.\label{eq: field 1}
\end{equation}
Radiation reaction requires the formulation of $\mathcal{F}_{(\mathring{\pm})}(\hat{x}(\tau,\omega))$.
We consider a decomposition of its first two leading-order terms for
an easier analysis of $G_{\mathrm{(\mathring{\pm})}}(x,\hat{x}(\tau',\bullet))|_{x=\hat{x}(\tau,\omega)}$:
\begin{equation}
\mathcal{F}_{(\mathring{\pm})}(\hat{x}(\tau,\omega))=\mathcal{F}_{(\mathring{\pm})}(\mathbb{E}\llbracket\hat{x}(\tau,\bullet)\rrbracket)+\delta\hat{x}^{\alpha}(\tau,\omega)\cdot\partial_{\alpha}\mathcal{F}_{(\mathring{\pm})}(\mathbb{E}\llbracket\hat{x}(\tau,\bullet)\rrbracket)+O(\otimes^{2}\delta\hat{x}(\tau,\omega))\label{eq: field-expansion}
\end{equation}
The details of the each terms of $\mathcal{F}_{(\mathring{\pm})}(\mathbb{E}\llbracket\hat{x}(\tau,\bullet)\rrbracket)$
and $\delta\hat{x}^{\alpha}(\tau,\omega)\cdot\partial_{\alpha}\mathcal{F}_{(\mathring{\pm})}(\mathbb{E}\llbracket\hat{x}(\tau,\bullet)\rrbracket)$
are discussed in the following small section.

\subsection{Radiation fields}

\noindent Consider $\mathcal{A}_{(\mathring{\pm})}(\mathbb{E}\llbracket\hat{x}(\tau,\bullet)\rrbracket)$
as the leading order term of $\mathcal{A}_{(\mathring{\pm})}(\hat{x}(\tau,\omega))$. 
\begin{lem}
\begin{leftbar}\label{sep_probable}For $\varOmega_{\tau}^{\mathrm{ave}}\coloneqq\{\omega|\hat{x}(\tau,\omega)=\mathbb{E}\llbracket\hat{x}(\tau,\bullet)\rrbracket\}\subset\varOmega$,
\begin{equation}
p(x,\tau)=\mathscr{P}(\varOmega_{\tau}^{\mathrm{ave}})\times\delta^{4}(x-\mathbb{E}\llbracket\hat{x}(\tau,\bullet)\rrbracket)+\int_{\varOmega\backslash\varOmega_{\tau}^{\mathrm{ave}}}d\mathscr{P}(\omega)\delta^{4}(x-\hat{x}(\tau,\omega))\,.
\end{equation}
\end{leftbar}
\end{lem}
In order to this lemma, $\mathcal{A}_{(\mathring{\pm})}(x)$ is separated
into two terms:
\begin{eqnarray}
\mathcal{A}_{(\mathring{\pm})}^{\mu}(x) & = & -\frac{e}{c\varepsilon_{0}}\int_{\mathbb{R}}d\tau'\mathscr{P}(\varOmega_{\tau'}^{\mathrm{ave}})\mathrm{Re}\left\{ \mathcal{V}^{\nu}(\mathbb{E}\llbracket\hat{x}(\tau',\bullet)\rrbracket)\right\} G_{\mathrm{(\mathring{\pm})}}(x,\mathbb{E}\llbracket\hat{x}(\tau',\bullet)\rrbracket)\nonumber \\
 &  & -\frac{e}{c\varepsilon_{0}}\int_{\mathbb{R}}d\tau'\int_{\varOmega\backslash\varOmega_{\tau'}^{\mathrm{ave}}}d\mathscr{P}(\omega')\mathrm{Re}\left\{ \mathcal{V}^{\nu}(\hat{x}(\tau',\omega'))\right\} G_{\mathrm{(\mathring{\pm})}}(x,\hat{x}(\tau',\omega'))
\end{eqnarray}
$G_{\mathrm{(\mathring{\pm})}}(x,\hat{x}(\tau',\omega'))$ includes
the light-cone condition via $\delta\left([x_{\mu}-\hat{x}_{\mu}(\tau',\omega')]\cdot[x^{\mu}-\hat{x}^{\mu}(\tau',\omega')]\right)$
in its definition. Therefore, the substitution $x=\mathbb{E}\llbracket\hat{x}(\tau,\bullet)\rrbracket$
for the derivation of $\mathcal{A}_{(\mathring{\pm})}(\mathbb{E}\llbracket\hat{x}(\tau,\bullet)\rrbracket)$
imposes
\begin{equation}
\left[\mathbb{E}\llbracket\hat{x}_{\mu}(\tau,\bullet)\rrbracket-\hat{x}_{\mu}(\tau',\omega')\right]\cdot\left[\mathbb{E}\llbracket\hat{x}^{\mu}(\tau,\bullet)\rrbracket-\hat{x}^{\mu}(\tau',\omega')\right]=0\,.\label{eq: light-cone}
\end{equation}
We regard this equation as an electron emits a field at $\hat{x}(\tau',\omega')$
and the both of the electron and the field interact at $\mathbb{E}\llbracket\hat{x}(\tau,\bullet)\rrbracket$
after their propagations. Consider the case $\omega'\in\varOmega\backslash\varOmega_{\tau'}^{\mathrm{ave}}$.
In this case, $\hat{x}(\tau',\omega')$ and $\mathbb{E}\llbracket\hat{x}(\tau,\bullet)\rrbracket$
locate at the separated points in the Minkowski spacetime even if
$\tau'=\tau$. Some of $(\tau',\omega')$ may hold the light-cone
condition of (\ref{eq: light-cone}). It expresses the self-interaction
with its finite size correction\footnote{In this diagram, the wiggly line and double line represent a photon
and an electron dressing an external field like a laser. The Green
function regulates the propagation of a photon and its two indexes
are the positions on a trajectory of an electron.}: 

\begin{singlespace}
\noindent \begin{center}
\begin{fmffile}{diagram}
$(\mathrm{self\mathchar`-interaction})=\begin{gathered}
\begin{fmfgraph}(45,15)
  \fmfleft{i1}
  \fmfright{o1}
  \fmf{double_arrow}{i1,v1}
  \fmf{double_arrow}{v1,v2}
  \fmf{photon,left,tension=0}{v1,v2}
  \fmf{double_arrow}{v2,o1}
  \fmfdotn{v}{2}
  \fmfblob{0.075w}{v1}
\end{fmfgraph}
\end{gathered}$
\end{fmffile}
\par\end{center}
\end{singlespace}

\noindent The finite size correction is denoted by the blob as the
set of $\{\hat{x}(\tau',\omega')|\omega'\in\varOmega\backslash\varOmega_{\tau'}^{\mathrm{ave}}\}$.
This condition appears with the probability of $\mathscr{P}(\varOmega\backslash\varOmega_{\tau'}^{\mathrm{ave}})$,
without any energy-momentum transfer on an electron between at the
initial and at final states. Namely, 
\[
-\frac{e}{c\varepsilon_{0}}\int_{\mathbb{R}}d\tau'\int_{\varOmega\backslash\varOmega_{\tau'}^{\mathrm{ave}}}d\mathscr{P}(\omega')\,\mathrm{Re}\left\{ \mathcal{V}^{\nu}(\hat{x}(\tau',\omega'))\right\} \left.G_{\mathrm{(\mathring{\pm})}}(x,\hat{x}(\tau',\omega'))\right|_{x=\mathbb{E}\llbracket\hat{x}(\tau,\bullet)\rrbracket}
\]
has to be renormalized as a part of its Coulomb field attaching to
an electron. Let us recall (\ref{eq: LA}) of the naive idea of an
observable energy-momentum transfer via radiation. Such radiation
can be found with the probability of $\mathscr{P}(\varOmega_{\tau'}^{\mathrm{ave}})=1-\mathscr{P}(\varOmega\backslash\varOmega_{\tau'}^{\mathrm{ave}})$
in the present model. Hence, $\omega'\in\varOmega_{\tau'}^{\mathrm{ave}}$
represents the sample paths for the derivation of radiation reaction
\footnote{Consider $\mathcal{A}_{(\mathrm{effective},\mathring{\pm})}^{\nu}(\mathbb{E}\llbracket\hat{x}(\tau,\bullet)\rrbracket)$
as $O(\mathbb{E}\llbracket\overset{0}{\otimes}\delta\hat{x}(\tau,\bullet)\rrbracket)$
in $\mathcal{A}_{(\mathrm{effective},\mathring{\pm})}^{\nu}(\hat{x}(\tau,\omega))$.}. 
\begin{eqnarray}
\mathcal{A}_{(\mathrm{effective},\mathring{\pm})}^{\nu}(\mathbb{E}\llbracket\hat{x}(\tau,\bullet)\rrbracket) & \coloneqq & -\frac{e}{c\varepsilon_{0}}\int_{\mathbb{R}}d\tau'\mathscr{P}(\varOmega_{\tau'}^{\mathrm{ave}})\mathrm{Re}\left\{ \mathcal{V}^{\nu}(\mathbb{E}\llbracket\hat{x}(\tau',\bullet)\rrbracket)\right\} \left.G_{(\mathring{\pm})}(x,\mathbb{E}\llbracket\hat{x}(\tau',\bullet)\rrbracket)\right|_{x=\mathbb{E}\llbracket\hat{x}(\tau,\bullet)\rrbracket}\,.
\end{eqnarray}

\begin{lem}
\begin{leftbar}\label{semi-classical-relation}Consider a D-progressive
$\hat{x}(\circ,\bullet)$ and $\delta\hat{x}(\tau,\omega)\coloneqq\hat{x}(\tau,\omega)-\mathbb{E}\llbracket\hat{x}(\tau,\bullet)\rrbracket$.
Then, the following relations are fulfilled:
\begin{equation}
\mathrm{Re}\left\{ \mathcal{V}^{\alpha}(\mathbb{E}\llbracket\hat{x}(\tau,\bullet)\rrbracket)\right\} -\frac{d\mathbb{E}\llbracket\hat{x}^{\alpha}(\tau,\bullet)\rrbracket}{d\tau}=O(\mathbb{E}\llbracket\overset{2}{\otimes}\delta\hat{x}(\tau,\bullet)\rrbracket)\label{eq: relation-001-1}
\end{equation}
\begin{equation}
\frac{d\mathbb{E}\llbracket\hat{x}_{\alpha}(\tau,\bullet)\rrbracket}{d\tau}\cdot\frac{d\mathbb{E}\llbracket\hat{x}^{\alpha}(\tau,\bullet)\rrbracket}{d\tau}-c^{2}=O(\mathbb{E}\llbracket\overset{2}{\otimes}\delta\hat{x}(\tau,\bullet)\rrbracket)\label{eq: relation-002-1}
\end{equation}
\end{leftbar}\end{lem}
\begin{proof}
By Nelson's partial integral formula ({\bf Lemma 31} in Ref.\cite{Seto Volume I}),
$d/d\tau\mathbb{E}\llbracket\hat{x}^{\alpha}(\tau,\bullet)\rrbracket=\mathbb{E}\llbracket\mathrm{Re}\{\mathcal{V}^{\alpha}(\hat{x}(\tau,\bullet))\}\rrbracket$
is satisfied. Since $\mathbb{E}\llbracket\delta\hat{x}(\tau,\omega)\rrbracket=0$,
the first relation (\ref{eq: relation-001-1}) is derived as follows:
\begin{eqnarray}
\mathrm{Re}\{\mathcal{V}^{\alpha}(\mathbb{E}\llbracket\hat{x}(\tau,\bullet)\rrbracket)\}-\frac{d}{d\tau}\mathbb{E}\llbracket\hat{x}^{\alpha}(\tau,\bullet)\rrbracket & = & \mathbb{E}\llbracket\mathrm{Re}\{\mathcal{V}^{\alpha}(\mathbb{E}\llbracket\hat{x}(\tau,\bullet)\rrbracket)\}\rrbracket-\mathbb{E}\llbracket\mathrm{Re}\{\mathcal{V}^{\alpha}(\hat{x}(\tau,\bullet))\}\rrbracket\nonumber \\
 & = & \mathbb{E}\llbracket\mathrm{Re}\{\mathcal{V}^{\alpha}(\mathbb{E}\llbracket\hat{x}(\tau',\bullet)\rrbracket)\}-\mathrm{Re}\{\mathcal{V}^{\alpha}(\hat{x}(\tau,\bullet))\}\rrbracket\nonumber \\
 & = & \mathbb{E}\llbracket-\delta\hat{x}^{\beta}(\tau,\bullet)\cdot\partial_{\beta}\mathrm{Re}\{\mathcal{V}^{\alpha}(\mathbb{E}\llbracket\hat{x}(\tau,\bullet)\rrbracket)\}+O(\overset{2}{\otimes}\delta\hat{x}(\tau,\bullet))\rrbracket\nonumber \\
 & = & O(\mathbb{E}\llbracket\overset{2}{\otimes}\delta\hat{x}(\tau,\bullet)\rrbracket)
\end{eqnarray}
The second relation (\ref{eq: relation-002-1}) is demonstrated by
\begin{equation}
\frac{d}{d\tau}\mathbb{E}\llbracket\hat{x}_{\alpha}(\tau,\bullet)\rrbracket\cdot\frac{d}{d\tau}\mathbb{E}\llbracket\hat{x}^{\alpha}(\tau,\bullet)\rrbracket-\mathbb{E}\left\llbracket \mathcal{V}_{\alpha}^{*}(\hat{x}(\tau,\bullet))\mathcal{V}^{\alpha}(\hat{x}(\tau,\bullet))\right\rrbracket =O(\mathbb{E}\llbracket\overset{2}{\otimes}\delta\hat{x}(\tau,\bullet)\rrbracket)
\end{equation}
with the help of the Lorentz invariant $\mathbb{E}\llbracket\mathcal{V}_{\alpha}^{*}(\hat{x}(\tau,\bullet))\mathcal{V}^{\alpha}(\hat{x}(\tau,\bullet))\rrbracket=c^{2}$
of {\bf Lemma 29} in \cite{Seto Volume I}.
\end{proof}
Transforming $\mathcal{A}_{(\mathrm{effective},\mathring{\pm})}^{\nu}(\mathbb{E}\llbracket\hat{x}(\tau,\bullet)\rrbracket)$
by the above, 
\begin{eqnarray}
\mathcal{A}_{(\mathrm{effective},\mathring{\pm})}^{\nu}(\mathbb{E}\llbracket\hat{x}(\tau,\bullet)\rrbracket) & \coloneqq & -\frac{e}{c\varepsilon_{0}}\int_{\mathbb{R}}d\tau'\,\mathscr{P}(\varOmega_{\tau'}^{\mathrm{ave}})\frac{d\mathbb{E}\llbracket\hat{x}^{\nu}(\tau',\bullet)\rrbracket}{d\tau'}\left.G_{(\mathring{\pm})}(x,\mathbb{E}\llbracket\hat{x}(\tau',\bullet)\rrbracket)\right|_{x=\mathbb{E}\llbracket\hat{x}(\tau,\bullet)\rrbracket}\,.
\end{eqnarray}

\noindent Thus, its field strength is
\begin{multline}
\mathcal{F}_{(\mathrm{effective},\mathring{\pm})}^{\mu\nu}(\mathbb{E}\llbracket\hat{x}(\tau,\bullet)\rrbracket)=-\frac{e}{c\varepsilon_{0}}\int_{\mathbb{R}}d\tau'\,\mathscr{P}(\varOmega_{\tau'}^{\mathrm{ave}})\\
\left(\begin{gathered}\frac{d\mathbb{E}\llbracket\hat{x}^{\nu}(\tau',\bullet)\rrbracket}{d\tau'}\cdot\partial^{\mu}-\frac{d\mathbb{E}\llbracket\hat{x}^{\mu}(\tau',\bullet)\rrbracket}{d\tau'}\cdot\partial^{\nu}\end{gathered}
\right)\left.G_{(\mathring{\pm})}(x,\mathbb{E}\llbracket\hat{x}(\tau',\bullet)\rrbracket)\right|_{x=\mathbb{E}\llbracket\hat{x}(\tau,\bullet)\rrbracket}\,.\label{eq: F_ave-001}
\end{multline}
Let us follow the method in Ref.\cite{Hartemann(2001)}. By {\bf Lemma \ref{semi-classical-relation}},
the Green function $G_{(\mathring{\pm})}$ is transformed as
\begin{eqnarray}
G_{(\mathring{\pm})}(\mathbb{E}\llbracket\hat{x}(\tau,\bullet)\rrbracket,\mathbb{E}\llbracket\hat{x}(\tau',\bullet)\rrbracket) & = & \frac{\theta(\mathring{\pm}\mathbb{E}\llbracket\hat{x}^{0}(\tau,\bullet)-\hat{x}^{0}(\tau',\bullet)\rrbracket)\times\delta(\tau-\tau')}{4\pi c^{2}\times|\tau'-\tau|}\nonumber \\
 &  & \qquad\times\left[1+O\left((\tau'-\tau)^{2},\mathbb{E}\llbracket\overset{2}{\otimes}\delta\hat{x}(\tau,\omega)\rrbracket\right)\right]\,.\label{eq: def_Green}
\end{eqnarray}
Where, the causality is introduced by $\theta(\mathring{\pm}\mathbb{E}\llbracket\hat{x}^{0}(\tau,\bullet)-\hat{x}^{0}(\tau',\bullet)\rrbracket)=\theta(\mathring{\pm}(\tau-\tau'))$
and (\ref{eq: relation-002-1}) is employed. In addition, 
\begin{eqnarray}
\partial^{\mu}\left.G_{(\mathring{\pm})}(x,\mathbb{E}\llbracket\hat{x}(\tau',\bullet)\rrbracket)\right|_{x=\mathbb{E}\llbracket\hat{x}(\tau,\bullet)\rrbracket} & = & -\frac{\mathbb{E}\llbracket\hat{x}^{\mu}(\tau',\bullet)-\hat{x}^{\mu}(\tau,\bullet)\rrbracket}{c^{2}\times(\tau'-\tau)}\frac{d}{d\tau'}G_{(\mathring{\pm})}(\mathbb{E}\llbracket\hat{x}(\tau,\bullet)\rrbracket,\mathbb{E}\llbracket\hat{x}(\tau',\bullet)\rrbracket)\nonumber \\
 &  & \qquad\times\left[1+O\left((\tau'-\tau)^{2},\mathbb{E}\llbracket\overset{2}{\otimes}\delta\hat{x}(\tau,\omega)\rrbracket\right)\right]\,.\label{eq:partial_G}
\end{eqnarray}
$\mathcal{F}_{(\mathrm{effective},\mathring{\pm})}(\mathbb{E}\llbracket\hat{x}(\tau,\bullet)\rrbracket)$
of (\ref{eq: F_ave-001}) becomes 
\begin{eqnarray}
\mathcal{F}_{(\mathrm{effective},\mathring{\pm})}(\mathbb{E}\llbracket\hat{x}(\tau,\bullet)\rrbracket) & = & \frac{e}{c\varepsilon_{0}}\int_{\mathbb{R}}d\tau'\,\left[\begin{gathered}\left(\mathscr{P}(\varOmega_{\tau'}^{\mathrm{ave}})\cdot\frac{d\mathbb{E}\llbracket\hat{x}^{\nu}(\tau',\bullet)\rrbracket}{d\tau'}\right)\times\mathbb{E}\llbracket\hat{x}^{\mu}(\tau,\bullet)-\hat{x}^{\mu}(\tau',\bullet)\rrbracket\\
-\left(\mathscr{P}(\varOmega_{\tau'}^{\mathrm{ave}})\cdot\frac{d\mathbb{E}\llbracket\hat{x}^{\mu}(\tau',\bullet)\rrbracket}{d\tau'}\right)\times\mathbb{E}\llbracket\hat{x}^{\nu}(\tau,\bullet)-\hat{x}^{\nu}(\tau',\bullet)\rrbracket
\end{gathered}
\right]\nonumber \\
 &  & \times\frac{\mathbb{E}\llbracket\hat{x}^{\mu}(\tau',\bullet)-\hat{x}^{\mu}(\tau,\bullet)\rrbracket}{c^{2}\times(\tau'-\tau)}\frac{d}{d\tau'}G_{(\mathring{\pm})}(\mathbb{E}\llbracket\hat{x}(\tau,\bullet)\rrbracket,\mathbb{E}\llbracket\hat{x}(\tau',\bullet)\rrbracket)\nonumber \\
 &  & \times\left[1+O\left((\tau'-\tau)^{2}\right)\right]\,.
\end{eqnarray}
Since
\begin{eqnarray}
\lefteqn{{\left(\mathscr{P}(\varOmega_{\tau'}^{\mathrm{ave}})\cdot\frac{d}{d\tau}\mathbb{E}\llbracket\hat{x}^{\nu}(\tau',\bullet)\rrbracket\right)\cdot\mathbb{E}\llbracket\hat{x}^{\mu}(\tau,\bullet)-\hat{x}^{\mu}(\tau',\bullet)\rrbracket-(\mu\longleftrightarrow\nu)}}\nonumber \\
 & = & -(\tau-\tau')^{2}\times\frac{\mathscr{P}(\varOmega_{\tau}^{\mathrm{ave}})}{2}\times\left[\frac{d\mathbb{E}\left\llbracket \hat{x}^{\mu}(\tau,\bullet)\right\rrbracket }{d\tau}\cdot\frac{d^{2}\mathbb{E}\left\llbracket \hat{x}^{\nu}(\tau,\bullet)\right\rrbracket }{d\tau^{2}}-(\mu\longleftrightarrow\nu)\right]\nonumber \\
 &  & -(\tau-\tau')^{3}\times\mathscr{P}(\varOmega_{\tau}^{\mathrm{ave}})\times\left[\begin{gathered}\frac{1}{2}\frac{\ln\mathscr{P}(\varOmega_{\tau}^{\mathrm{ave}})}{d\tau}\times\frac{d\mathbb{E}\left\llbracket \hat{x}^{\mu}(\tau,\bullet)\right\rrbracket }{d\tau}\cdot\frac{d^{2}\mathbb{E}\left\llbracket \hat{x}^{\nu}(\tau,\bullet)\right\rrbracket }{d\tau^{2}}\\
\frac{1}{3}\times\frac{d\mathbb{E}\left\llbracket \hat{x}^{\mu}(\tau,\bullet)\right\rrbracket }{d\tau}\cdot\frac{d^{3}\mathbb{E}\left\llbracket \hat{x}^{\nu}(\tau,\bullet)\right\rrbracket }{d\tau^{3}}-(\mu\longleftrightarrow\nu)
\end{gathered}
\right]\nonumber \\
 &  & +O\left((\tau'-\tau)^{4},\mathbb{E}\llbracket\overset{2}{\otimes}\delta\hat{x}(\tau,\omega)\rrbracket\right)\,,
\end{eqnarray}
 the field $\mathcal{F}_{(\mathrm{effective},\mathring{\pm})}(\mathbb{E}\llbracket\hat{x}(\tau,\bullet)\rrbracket)$
is obtained:
\begin{eqnarray}
\mathcal{F}_{(\mathrm{effective},\mathring{\pm})}^{\mu\nu}(\mathbb{E}\llbracket\hat{x}(\tau,\bullet)\rrbracket) & = & \frac{3}{4}\frac{m_{0}\tau_{0}\mathscr{P}(\varOmega_{\tau}^{\mathrm{ave}})}{ec^{2}}\times\left(\begin{gathered}\frac{d^{2}\mathbb{\mathbb{E}}\llbracket\hat{x}^{\mu}(\tau,\bullet)\rrbracket}{d\tau^{2}}\cdot\frac{d\mathbb{E}\llbracket\hat{x}^{\nu}(\tau,\bullet)\rrbracket}{d\tau}\\
-\frac{d^{2}\mathbb{E}\llbracket\hat{x}^{\nu}(\tau,\bullet)\rrbracket}{d\tau^{2}}\cdot\frac{d\mathbb{\mathbb{E}}\llbracket\hat{x}^{\mu}(\tau,\bullet)\rrbracket}{d\tau}
\end{gathered}
\right)\int_{\mathbb{R}}d\tau'\,\frac{\delta(\tau'-\tau)}{|\tau'-\tau|}\nonumber \\
 &  & \mathring{\mp}\frac{m_{0}\tau_{0}\mathscr{P}(\varOmega_{\tau}^{\mathrm{ave}})}{ec^{2}}\times\left(\begin{gathered}\dot{a}(\mathbb{\mathbb{E}}\llbracket\hat{x}(\tau,\bullet)\rrbracket)\otimes\frac{d\mathbb{E}\llbracket\hat{x}(\tau,\bullet)\rrbracket}{d\tau}\\
-\frac{d\mathbb{E}\llbracket\hat{x}(\tau,\bullet)\rrbracket}{d\tau}\otimes\dot{a}(\mathbb{\mathbb{E}}\llbracket\hat{x}(\tau,\bullet)\rrbracket)
\end{gathered}
\right)
\end{eqnarray}
Where, $\tau_{0}\coloneqq e^{2}/6\pi\varepsilon_{0}m_{0}c^{3}$ and
$\dot{a}(\mathbb{\mathbb{E}}\llbracket\hat{x}(\tau,\bullet)\rrbracket)\in\mathbb{V}_{\mathrm{M}}^{\mathrm{4}}$
denotes
\begin{equation}
\dot{a}(\mathbb{\mathbb{E}}\llbracket\hat{x}(\tau,\bullet)\rrbracket)\coloneqq\frac{d^{3}\mathbb{E}\llbracket\hat{x}(\tau,\bullet)\rrbracket}{d\tau^{3}}+\frac{3}{2}\frac{d\ln\mathscr{P}(\varOmega_{\tau}^{\mathrm{ave}})}{d\tau}\frac{d^{2}\mathbb{E}\llbracket\hat{x}(\tau,\bullet)\rrbracket}{d\tau^{2}}\,.\label{eq: dot_a}
\end{equation}
Therefore, the effective radiation reaction field on $\mathbb{E}\llbracket\hat{x}(\circ,\bullet)\rrbracket$
is the following: 
\begin{lem}[Effective radiation reaction field 1]
\label{F_average}\begin{leftbar}The effective radiation reaction
field $\mathfrak{F}(\mathbb{E}\llbracket\hat{x}(\tau,\bullet)\rrbracket)\in\mathbb{V}_{\mathrm{M}}^{\mathrm{4}}\otimes\mathbb{V}_{\mathrm{M}}^{\mathrm{4}}$
is defined as follows:
\begin{eqnarray}
\mathfrak{F}(\mathbb{E}\llbracket\hat{x}(\tau,\bullet)\rrbracket) & \coloneqq & \frac{\mathcal{F}_{(\mathrm{effective},\mathring{+})}(\mathbb{E}\llbracket\hat{x}(\tau,\bullet)\rrbracket)-\mathcal{F}_{(\mathrm{effective},\mathring{-})}(\mathbb{E}\llbracket\hat{x}(\tau,\bullet)\rrbracket)]}{2}\nonumber \\
 & = & -\frac{m_{0}\tau_{0}\mathscr{P}(\varOmega_{\tau}^{\mathrm{ave}})}{ec^{2}}\times\left(\begin{gathered}\dot{a}(\mathbb{\mathbb{E}}\llbracket\hat{x}(\tau,\bullet)\rrbracket)\otimes\frac{d\mathbb{E}\llbracket\hat{x}(\tau,\bullet)\rrbracket}{d\tau}-\frac{d\mathbb{E}\llbracket\hat{x}(\tau,\bullet)\rrbracket}{d\tau}\otimes\dot{a}(\mathbb{\mathbb{E}}\llbracket\hat{x}(\tau,\bullet)\rrbracket)\end{gathered}
\right)\label{eq: RR field on ave-point}
\end{eqnarray}
\end{leftbar}
\end{lem}
Consider the following transformation of (\ref{eq: RR field on ave-point}),
the wider relation is found. 
\begin{eqnarray}
\int_{\varOmega_{\tau}^{\mathrm{ave}}}d\mathscr{P}(\omega)\mathfrak{F}(\hat{x}(\tau,\omega)) & = & \int_{\varOmega_{\tau}^{\mathrm{ave}}}d\mathscr{P}(\omega)\left[-\frac{m_{0}\tau_{0}}{ec^{2}}\times\left(\begin{gathered}\dot{a}(\hat{x}(\tau,\omega))\otimes\mathrm{Re}\left\{ \mathcal{V}(\hat{x}(\tau,\omega))\right\} \\
-\mathrm{Re}\left\{ \mathcal{V}(\hat{x}(\tau,\omega))\right\} \otimes\dot{a}(\hat{x}(\tau,\omega))
\end{gathered}
\right)\right]
\end{eqnarray}
\begin{snugshade}
\begin{lem}[Effective radiation reaction field 2]
\label{F_full}\begin{leftbar}The field $\mathfrak{F}(\hat{x}(\tau,\omega))\in\mathbb{V}_{\mathrm{M}}^{\mathrm{4}}\otimes\mathbb{V}_{\mathrm{M}}^{\mathrm{4}}$
is derived as follows:
\begin{equation}
\mathfrak{F}(\hat{x}(\tau,\omega))=-\frac{m_{0}\tau_{0}}{ec^{2}}\times\left[\begin{gathered}\dot{a}(x)\otimes\mathrm{Re}\left\{ \mathcal{V}(x)\right\} -\mathrm{Re}\left\{ \mathcal{V}(x)\right\} \otimes\dot{a}(x)\end{gathered}
\right]_{x=\hat{x}(\tau,\omega)}\label{eq: F_Full}
\end{equation}
Where, $\dot{a}(x)$ satisfies 
\begin{equation}
\dot{a}(\mathbb{\mathbb{E}}\llbracket\hat{x}(\tau,\bullet)\rrbracket)\coloneqq\frac{d^{3}\mathbb{\mathbb{E}}\llbracket\hat{x}(\tau,\bullet)\rrbracket}{d\tau^{3}}+\frac{3}{2}\frac{d\ln\mathscr{P}(\varOmega_{\tau}^{\mathrm{ave}})}{d\tau}\times\frac{d^{2}\mathbb{\mathbb{E}}\llbracket\hat{x}(\tau,\bullet)\rrbracket}{d\tau^{2}}+O(\mathbb{E}\llbracket\overset{2}{\otimes}\delta\hat{x}(\tau,\omega)\rrbracket)\,.\label{eq: dot-a-2}
\end{equation}
\end{leftbar}
\end{lem}
\end{snugshade}{\bf Lemma \ref{F_average}} is derived from {\bf Lemma \ref{F_full}}
by 
\begin{equation}
\mathfrak{F}(\mathbb{E}\llbracket\hat{x}(\tau,\bullet)\rrbracket)=\int_{\varOmega_{\tau}^{\mathrm{ave}}}d\mathscr{P}(\omega)\,\mathfrak{F}(\hat{x}(\tau,\omega))\,.
\end{equation}
Since 
\begin{eqnarray}
\partial_{\alpha}\mathfrak{F}(\mathbb{E}\llbracket\hat{x}(\tau,\bullet)\rrbracket) & = & -\frac{m_{0}\tau_{0}}{ec^{2}}\int_{\varOmega_{\tau}^{\mathrm{ave}}}d\mathscr{P}(\omega)\partial_{\alpha}\left[\begin{gathered}\dot{a}(x)\otimes\mathrm{Re}\left\{ \mathcal{V}(x)\right\} -\mathrm{Re}\left\{ \mathcal{V}(x)\right\} \otimes\dot{a}(x)\end{gathered}
\right]_{x=\hat{x}(\tau,\omega)}\nonumber \\
 & = & -\frac{m_{0}\tau_{0}\mathscr{P}(\varOmega_{\tau}^{\mathrm{ave}})}{ec^{2}}\partial_{\alpha}\left[\begin{gathered}\dot{a}(x)\otimes\mathrm{Re}\left\{ \mathcal{V}(x)\right\} -\mathrm{Re}\left\{ \mathcal{V}(x)\right\} \otimes\dot{a}(x)\end{gathered}
\right]_{x=\hat{x}(\tau,\omega)}\,,
\end{eqnarray}
the leading-order terms of 
\begin{equation}
\mathfrak{F}(\hat{x}(\tau,\omega))=\mathfrak{F}(\mathbb{E}\llbracket\hat{x}(\tau,\bullet)\rrbracket)+\delta\hat{x}^{\alpha}(\tau,\omega)\cdot\partial_{\alpha}\mathfrak{F}(\mathbb{E}\llbracket\hat{x}(\tau,\bullet)\rrbracket)+O(\otimes^{2}\delta\hat{x}(\tau,\omega))
\end{equation}
can be estimated. By putting $F_{\mathrm{ex}}$ a non-trivial solution
of (\ref{eq: Maxwell_for_derivation}) such that $\partial_{\mu}F_{\mathrm{ex}}^{\mu\nu}(x)=0$,
the particular solution $\mathcal{F}$ becomes as follows:
\begin{equation}
\mathcal{F}(\hat{x}(\tau,\omega))=F_{\mathrm{ex}}(\hat{x}(\tau,\omega))+\mathfrak{F}(\hat{x}(\tau,\omega))+\delta\mathfrak{f}(\hat{x}(\tau,\omega))
\end{equation}
This is the suggestion of (\ref{eq: relation of radiation field})
on {\bf Lemma \ref{F_LAD_stochastic}}. Where, the singularity $\delta\mathfrak{f}(x)$,
namely, the part of non-effective radiation, is defined by
\begin{eqnarray}
\delta\mathfrak{f}^{\mu\nu}(x) & \coloneqq & -\frac{e}{c\varepsilon_{0}}\int_{\mathbb{R}}d\tau'\,\int_{\varOmega\backslash\varOmega_{\tau'}^{\mathrm{ave}}}d\mathscr{P}(\omega')\left(\begin{gathered}\mathrm{Re}\left\{ \mathcal{V}^{\nu}(\hat{x}(\tau',\omega'))\right\} \partial^{\mu}\\
-\mathrm{Re}\left\{ \mathcal{V}^{\mu}(\hat{x}(\tau',\omega'))\right\} \partial^{\nu}
\end{gathered}
\right)\frac{G_{\mathrm{(\mathring{+})}}(x,\hat{x}(\tau',\omega'))-G_{\mathrm{(\mathring{-})}}(x,\hat{x}(\tau',\omega'))}{2}\nonumber \\
 &  & -\frac{e}{c\varepsilon_{0}}\int_{\mathbb{R}}d\tau'\,\mathbb{E}\left\llbracket \left(\begin{gathered}\mathrm{Re}\left\{ \mathcal{V}^{\nu}(\hat{x}(\tau',\bullet))\right\} \partial^{\mu}\\
-\mathrm{Re}\left\{ \mathcal{V}^{\mu}(\hat{x}(\tau',\bullet))\right\} \partial^{\nu}
\end{gathered}
\right)\frac{G_{\mathrm{(\mathring{+})}}(x,\hat{x}(\tau',\bullet))+G_{\mathrm{(\mathring{-})}}(x,\hat{x}(\tau',\bullet))}{2}\right\rrbracket \,.
\end{eqnarray}
The appearance of the integral $-e/c\varepsilon_{0}\times\int_{\varOmega\backslash\varOmega_{\tau'}^{\mathrm{ave}}}d\mathscr{P}(\omega')\cdots$
in this equation is the difference from classical dynamics.

\subsection{Probability: $\mathscr{P}(\varOmega_{\tau}^{\mathrm{ave}})$}

Let us consider an equation of continuity for a D-progressive $\hat{x}(\circ,\bullet)$
on $(\varOmega,\mathcal{F},\mathscr{P})$, namely, 
\begin{equation}
\partial_{\tau}p(x,\tau)+\partial_{\mu}\left[\mathrm{Re}\{\mathcal{V}^{\mu}(x)\}p(x,\tau)\right]=0,\,\,x\in\hat{x}(\tau,\varOmega)\mathrm{\,\,for\,each\,}\tau\label{eq: eq of continuity}
\end{equation}
derived from the ''$\pm$''-Fokker-Planck equations (\ref{eq: Fokker-Planck})
\cite{Seto Volume I}. We want to evaluate the probability density
$p(\tau,x)$ by using this equation. By transforming it,
\begin{equation}
\partial_{\tau}\ln p(x,\tau)+\mathrm{Re}\{\mathcal{V}^{\mu}(x)\}\cdot\partial_{\mu}\ln p(x,\tau)+\partial_{\mu}\mathrm{Re}\{\mathcal{V}^{\mu}(x)\}=0,\,\,x\in\hat{x}(\tau,\varOmega)\mathrm{\,\,for\,each\,}\tau
\end{equation}
is found. Since 
\begin{equation}
\mathcal{V}^{\mu}(x)\coloneqq i\lambda^{2}\times\partial{}^{\alpha}\ln\phi(x)+\frac{e}{m_{0}}A{}^{\alpha}(x),\,\,x\in\hat{x}(\tau,\varOmega)\mathrm{\,\,for\,each\,}\tau,
\end{equation}
(\ref{eq: eq of continuity}) becomes
\begin{equation}
\partial_{\tau}\ln p(x,\tau)+\mathrm{Re}\{\mathcal{V}^{\mu}(x)\}\cdot\partial_{\mu}\ln p(x,\tau)=\mathrm{Re}\{\mathcal{V}^{\mu}(x)\}\cdot\partial_{\mu}\ln|\phi(x)|^{2}\,.
\end{equation}
Where, $\mathrm{Re}\{\mathcal{V}^{\mu}(x)\}=[\phi^{*}(x)\phi(x)]^{-1}\times j_{\mathrm{\mathrm{K\mathchar`-G}}}^{\mu}(x)$
and $\partial_{\mu}j_{\mathrm{\mathrm{K\mathchar`-G}}}^{\mu}(x)=0$
are employed (see Ref.\cite{Seto Volume I}).  The operator $\partial_{\tau}+\mathrm{Re}\{\mathcal{V}^{\mu}(x)\}\cdot\partial_{\mu}$
denotes the Lagrangian representation in fluid dynamics and $\partial_{\mu}\left[\phi^{*}(x)\phi(x)\right]\ne0$
on $x\in\hat{x}(\tau,\varOmega)$ for each $\tau$. Let us proceed
this idea aggressively. Consider it on $x\in\{\mathbb{E}\llbracket\hat{x}(\tau,\bullet)\rrbracket|\tau\in\mathbb{R}\}$,
then,
\begin{eqnarray}
\frac{d}{d\tau}\ln p(\mathbb{E}\llbracket\hat{x}(\tau,\bullet)\rrbracket,\tau) & = & \frac{d\mathbb{E}\llbracket\hat{x}^{\mu}(\tau,\bullet)\rrbracket}{d\tau}\cdot\partial_{\mu}\ln|\phi(\mathbb{E}\llbracket\hat{x}(\tau,\bullet)\rrbracket)|^{2}+O(\mathbb{E}\llbracket\overset{2}{\otimes}\delta\hat{x}(\tau,\bullet)\rrbracket)\nonumber \\
 & = & \frac{d}{d\tau}\ln|\phi(\mathbb{E}\llbracket\hat{x}(\tau,\bullet)\rrbracket)|^{2}+O(\mathbb{E}\llbracket\overset{2}{\otimes}\delta\hat{x}(\tau,\bullet)\rrbracket)\,.\label{eq: time_evo_prob}
\end{eqnarray}
When the RHS of (\ref{eq: time_evo_prob}) is well-defined, 
\begin{equation}
\ln\frac{p(\mathbb{E}\llbracket\hat{x}(\tau,\bullet)\rrbracket,\tau)}{p(\mathbb{E}\llbracket\hat{x}(0,\bullet)\rrbracket,0)}=\ln\frac{|\phi(\mathbb{E}\llbracket\hat{x}(\tau,\bullet)\rrbracket)|^{2}}{|\phi(\mathbb{E}\llbracket\hat{x}(0,\bullet)\rrbracket)|^{2}}+O(\mathbb{E}\llbracket\overset{2}{\otimes}\delta\hat{x}(\tau,\bullet)\rrbracket)\,.
\end{equation}

\begin{lem}
\begin{leftbar}The time-evolution of the probability $\mathscr{P}(\varOmega_{\tau}^{\mathrm{ave}})\coloneqq p(\mathbb{E}\llbracket\hat{x}(\tau,\bullet)\rrbracket,\tau)$
satisfying (\ref{eq: eq of continuity}) is bellow:
\begin{equation}
\mathscr{P}(\varOmega_{\tau}^{\mathrm{ave}})=\mathscr{P}(\varOmega_{0}^{\mathrm{ave}})\times\frac{|\phi(\mathbb{E}\llbracket\hat{x}(\tau,\bullet)\rrbracket)|^{2}}{|\phi(\mathbb{E}\llbracket\hat{x}(0,\bullet)\rrbracket)|^{2}}+O(\mathbb{E}\llbracket\overset{2}{\otimes}\delta\hat{x}(\tau,\bullet)\rrbracket)\label{eq: prob_evo}
\end{equation}
\end{leftbar}
\end{lem}

\section{Conclusion and discussion\label{Sect-Summary}}

We discussed the formulation of the kinematics and the dynamics of
a stochastic scalar electron equivalent to the Klein-Gordon particle
interacting with a field. Especially, we focused the derivation of
the radiation reaction effect as an application of {\bf Volume I}
\cite{Seto Volume I}.  Let us summarize the results of this article.\begin{snugshade}\begin{leftbar}
\begin{conclusion}[Radiation reaction]
\noindent \label{conclusion_1}For $\left(\mathit{\Omega},\mathcal{F},\mathscr{P}\right)$,
consider a kinematics of a D-progressive $\hat{x}(\circ,\bullet)$,
\begin{equation}
d\hat{x}(\tau,\omega)=\mathcal{V}_{\pm}(\hat{x}(\tau,\omega))d\tau+\lambda\times dW_{\pm}(\tau,\omega)
\end{equation}
on $(\mathbb{A}^{4}(\mathbb{V_{\mathrm{M}}^{\mathrm{4}}},g),\mathscr{B}(\mathbb{A}^{4}(\mathbb{V_{\mathrm{M}}^{\mathrm{4}}},g)),\mu)$.
This is coupled with the following dynamics for $\mathscr{\mathcal{V}}=(1-i)/2\times\mathcal{V}_{+}+(1+i)/2\times\mathcal{V}_{-}\in\mathbb{V_{\mathrm{M}}^{\mathrm{4}}}\oplus i\mathbb{V_{\mathrm{M}}^{\mathrm{4}}}$:
\begin{eqnarray}
m_{0}\mathfrak{D_{\tau}}\mathcal{V}^{\mu}(\hat{x}(\tau,\omega)) & = & -eF_{\mathrm{ex}}^{\mu\nu}(\hat{x}(\tau,\omega))\cdot\mathcal{V}_{\nu}(\hat{x}(\tau,\omega))\nonumber \\
 &  & +\frac{m_{0}\tau_{0}}{c^{2}}\times\left[\begin{gathered}\dot{a}^{\mu}(\hat{x}(\tau,\omega))\cdot\mathrm{Re}\left\{ \mathcal{V}^{\nu}(\hat{x}(\tau,\omega))\right\} \\
-\mathrm{Re}\left\{ \mathcal{V}^{\mu}(\hat{x}(\tau,\omega))\right\} \cdot\dot{a}^{\nu}(\hat{x}(\tau,\omega))
\end{gathered}
\right]\cdot\mathcal{V}_{\nu}(\hat{x}(\tau,\omega))\label{eq: EOM-final}
\end{eqnarray}

\noindent Where, $\dot{a}(\hat{x}(\tau,\omega))$ is defined by (\ref{eq: dot-a-2}).
This is equivalent to the Klein-Gordon equation with its radiation,
i.e., the quantization of LAD equation (\ref{eq:LAD}).
\end{conclusion}
\end{leftbar}\end{snugshade}The dynamics of a radiating stochastic
particle corresponding to the LAD equation was hereby derived, however,
(\ref{eq: EOM-final}) includes many higher order derivatives. Even
in the case of the LAD equation, the exponential factor $dv^{\mu}/d\tau\propto\exp(\tau/\tau_{0})$,
namely, the run-away problem remains in it \cite{Dirac(1938a)}. This
suffers us in the actual estimations and numerical simulations of
its experimental designs. Let us introduce the Ford-O'Connell \cite{Ford-O'Connell}/Landau-Lifshitz
\cite{Landau-Lifshitz} schemes for it, too. In the case of the Landau-Lifshitz
method, it is imposed by the following perturbation of $m_{0}\tau_{0}/c^{2}\times(d^{2}v^{\mu}/d\tau^{2}\cdot v^{\nu}-d^{2}v^{\nu}/d\tau^{2}\cdot v^{\mu})v_{\nu}$
in the LAD equation (\ref{eq:LAD}) with respect to $\tau_{0}=O(10^{-24}\mathrm{sec})$:
\begin{equation}
\frac{dv^{\mu}}{d\tau}=-\frac{e}{m_{0}}F_{\mathrm{ex}}^{\mu\nu}v_{\nu}+O(\tau_{0})
\end{equation}
\begin{equation}
\frac{d^{2}v^{\mu}}{d\tau^{2}}=-\frac{e}{m_{0}}\partial_{\alpha}F_{\mathrm{ex}}^{\mu\nu}\cdot v^{\alpha}v_{\nu}+\frac{e^{2}}{m_{0}^{2}}g_{\alpha\beta}F_{\mathrm{ex}}^{\mu\alpha}F_{\mathrm{ex}}^{\beta\nu}v_{\nu}+O(\tau_{0})
\end{equation}
Then, the Landau-Lifshitz equation is derived:
\begin{equation}
m_{0}\frac{dv^{\mu}}{d\tau}(\tau)=-eF_{\mathrm{ex}}^{\mu\nu}(x(\tau))v_{\nu}(\tau)-eF_{\mathrm{LL}}^{\mu\nu}(x(\tau))v_{\nu}(\tau)+O(\tau_{0}^{2})\label{eq: LL-classic}
\end{equation}
\begin{equation}
F_{\mathrm{LL}}^{\mu\nu}(x(\tau))\coloneqq\tau_{0}v^{\alpha}(\tau)\partial_{\alpha}F_{\mathrm{ex}}^{\mu\nu}(x(\tau))-\frac{e\tau_{0}}{m_{0}c^{2}}\left(\delta_{\theta}^{\mu}\delta_{\epsilon}^{\nu}-\delta_{\theta}^{\nu}\delta_{\epsilon}^{\mu}\right)g_{\alpha\beta}F_{\mathrm{ex}}^{\theta\alpha}(x)F_{\mathrm{ex}}^{\beta\gamma}(x)v_{\gamma}(\tau)v^{\epsilon}(\tau)\label{eq: F_LL_classical}
\end{equation}
 For applying it to the present model, consider Ehrenfest's theorem
of (\ref{eq: EOM-final}) at first.
\begin{eqnarray}
m_{0}\frac{d^{2}\mathbb{E}\llbracket\hat{x}(\tau,\bullet)\rrbracket}{d\tau^{2}} & = & -eF_{\mathrm{ex}}^{\mu\nu}(\mathbb{E}\llbracket\hat{x}(\tau,\bullet)\rrbracket)\frac{d\mathbb{E}\llbracket\hat{x}_{\nu}(\tau,\bullet)\rrbracket}{d\tau}\nonumber \\
 &  & +\frac{m_{0}\tau_{0}}{c^{2}}\times\left[\begin{gathered}\dot{a}^{\mu}(\mathbb{\mathbb{E}}\llbracket\hat{x}(\tau,\bullet)\rrbracket)\cdot\frac{d\mathbb{E}\llbracket\hat{x}^{\nu}(\tau,\bullet)\rrbracket}{d\tau}\\
-\dot{a}^{\nu}(\mathbb{\mathbb{E}}\llbracket\hat{x}(\tau,\bullet)\rrbracket)\cdot\frac{d\mathbb{E}\llbracket\hat{x}^{\mu}(\tau,\bullet)\rrbracket}{d\tau}
\end{gathered}
\right]\frac{d\mathbb{E}\llbracket\hat{x}_{\nu}(\tau,\bullet)\rrbracket}{d\tau}\nonumber \\
 &  & +O(\mathbb{E}\llbracket\overset{2}{\otimes}\delta\hat{x}(\tau,\bullet)\rrbracket)
\end{eqnarray}
The readers have to compare this and the LAD equation (\ref{eq:LAD}).
Its perturbation for the Landau-Lifshitz scheme is imposed by 
\begin{eqnarray}
\frac{d^{2}\mathbb{\mathbb{E}}\llbracket\hat{x}^{\mu}(\tau,\bullet)\rrbracket}{d\tau^{2}} & = & -\frac{e}{m_{0}}\times F_{\mathrm{ex}}^{\mu\nu}(\mathbb{\mathbb{E}}\llbracket\hat{x}(\tau,\bullet)\rrbracket)\frac{d\mathbb{\mathbb{E}}\llbracket\hat{x}_{\nu}(\tau,\bullet)\rrbracket}{d\tau}+O(\tau_{0},\mathbb{E}\llbracket\overset{2}{\otimes}\delta\hat{x}(\tau,\bullet)\rrbracket)
\end{eqnarray}
and also
\begin{eqnarray}
\frac{d^{3}\mathbb{\mathbb{E}}\llbracket\hat{x}^{\mu}(\tau,\bullet)\rrbracket}{d\tau^{3}} & = & -\frac{e}{m_{0}}\times\frac{dF_{\mathrm{ex}}^{\mu\nu}(\mathbb{\mathbb{E}}\llbracket\hat{x}(\tau,\bullet)\rrbracket)}{d\tau}\frac{d\mathbb{\mathbb{E}}\llbracket\hat{x}_{\nu}(\tau,\bullet)\rrbracket}{d\tau}\nonumber \\
 &  & +\frac{e^{2}}{m_{0}^{2}}\times g_{\alpha\beta}F_{\mathrm{ex}}^{\mu\alpha}(\mathbb{\mathbb{E}}\llbracket\hat{x}(\tau,\bullet)\rrbracket)F_{\mathrm{ex}}^{\beta\nu}(\mathbb{\mathbb{E}}\llbracket\hat{x}(\tau,\bullet)\rrbracket)\frac{d\mathbb{\mathbb{E}}\llbracket\hat{x}_{\nu}(\tau,\bullet)\rrbracket}{d\tau}+O(\tau_{0},\mathbb{E}\llbracket\overset{2}{\otimes}\delta\hat{x}(\tau,\bullet)\rrbracket)\,.
\end{eqnarray}
Thus, the effective radiation field at $x=\mathbb{E}\llbracket\hat{x}(\tau,\bullet)\rrbracket$
becomes
\begin{equation}
\mathfrak{F}(\mathbb{E}\llbracket\hat{x}(\tau,\bullet)\rrbracket)=\mathscr{P}(\varOmega_{\tau}^{\mathrm{ave}})\times F_{\mathrm{LL}}(\mathbb{\mathbb{E}}\llbracket\hat{x}(\tau,\bullet)\rrbracket)+\frac{3}{2}\tau_{0}\frac{d\mathscr{P}(\varOmega_{\tau}^{\mathrm{ave}})}{d\tau}\times F_{\mathrm{ex}}(\mathbb{\mathbb{E}}\llbracket\hat{x}(\tau,\bullet)\rrbracket)+O(\tau_{0}^{2},\mathbb{E}\llbracket\overset{2}{\otimes}\delta\hat{x}(\tau,\bullet)\rrbracket)\,,
\end{equation}
where, $(x(\tau),v(\tau))$ in $F_{\mathrm{LL}}(x(\tau))$ of (\ref{eq: F_LL_classical})
is replaced by $(\mathbb{\mathbb{E}}\llbracket\hat{x}(\tau,\bullet)\rrbracket,d\mathbb{\mathbb{E}}\llbracket\hat{x}(\tau,\bullet)\rrbracket/d\tau)$.
Thus, equation of motion along $\mathbb{E}\llbracket\hat{x}(\tau,\bullet)\rrbracket$
is 
\begin{eqnarray}
m_{0}\frac{d^{2}\mathbb{E}\llbracket\hat{x}^{\mu}(\tau,\bullet)\rrbracket}{d\tau^{2}} & = & \left[1+\frac{3}{2}\tau_{0}\frac{d\mathscr{P}(\varOmega_{\tau}^{\mathrm{ave}})}{d\tau}\right]\times\left[-eF_{\mathrm{ex}}^{\mu\nu}(\mathbb{\mathbb{E}}\llbracket\hat{x}(\tau,\bullet)\rrbracket)\frac{d\mathbb{\mathbb{E}}\llbracket\hat{x}_{\nu}(\tau,\bullet)\rrbracket}{d\tau}\right]\nonumber \\
 &  & +\mathscr{P}(\varOmega_{\tau}^{\mathrm{ave}})\times\left[-eF_{\mathrm{LL}}^{\mu\nu}(\mathbb{\mathbb{E}}\llbracket\hat{x}(\tau,\bullet)\rrbracket)\cdot\frac{d\mathbb{E}\llbracket\hat{x}_{\nu}(\tau,\bullet)\rrbracket}{d\tau}\right]\nonumber \\
 &  & +O(\tau_{0}^{2},\mathbb{E}\llbracket\overset{2}{\otimes}\delta\hat{x}(\tau,\bullet)\rrbracket)\,.
\end{eqnarray}
By using
\begin{equation}
\int_{\varOmega_{\tau}^{\mathrm{ave}}}d\mathscr{P}(\omega)\,\mathfrak{F}(\hat{x}(\tau,\omega))=\int_{\varOmega_{\tau}^{\mathrm{ave}}}d\mathscr{P}(\omega)\left[F_{\mathrm{LL}}(\hat{x}(\tau,\omega))+\frac{3}{2}\tau_{0}\frac{d\ln\mathscr{P}(\varOmega_{\tau}^{\mathrm{ave}})}{d\tau}\times F_{\mathrm{ex}}(\hat{x}(\tau,\omega))\right]+O(\tau_{0}^{2})\,,
\end{equation}
 the following is the perturbed equation of (\ref{eq: EOM-final}):

\begin{snugshade}\begin{leftbar}
\begin{lem}[Landau-Lifshitz's approximation]
\noindent Consider the perturbation of (\ref{eq: EOM-final}) for
$\tau_{0}$,
\begin{eqnarray}
m_{0}\mathfrak{D_{\tau}}\mathcal{V}^{\mu}(\hat{x}(\tau,\omega)) & = & -e\left[1+\frac{3}{2}\tau_{0}\frac{d\ln\mathscr{P}(\varOmega_{\tau}^{\mathrm{ave}})}{d\tau}\right]F_{\mathrm{ex}}^{\mu\nu}(\hat{x}(\tau,\omega))\cdot\mathcal{V}_{\nu}(\hat{x}(\tau,\omega))\nonumber \\
 &  & -eF_{\mathrm{LL}}^{\mu\nu}(\hat{x}(\tau,\omega))\cdot\mathcal{V}_{\nu}(\hat{x}(\tau,\omega))+O(\tau_{0}^{2})\label{eq: LL-stochastic}
\end{eqnarray}

\noindent where, $F_{\mathrm{LL}}^{\mu\nu}(\hat{x}(\tau,\omega))$
is defined by 
\begin{equation}
F_{\mathrm{LL}}^{\mu\nu}(x)=\tau_{0}\mathrm{Re}\left\{ \mathcal{V}^{\alpha}(x)\right\} \partial_{\alpha}F_{\mathrm{ex}}^{\mu\nu}(x)-\frac{e\tau_{0}}{m_{0}c^{2}}\left(\delta_{\theta}^{\mu}\delta_{\epsilon}^{\nu}-\delta_{\theta}^{\nu}\delta_{\epsilon}^{\mu}\right)g_{\alpha\beta}F_{\mathrm{ex}}^{\theta\alpha}(x)F_{\mathrm{ex}}^{\beta\gamma}(x)\mathrm{Re}\left\{ \mathcal{V}_{\gamma}(x)\right\} \mathrm{Re}\left\{ \mathcal{V}^{\epsilon}(x)\right\} \,.
\end{equation}
(\ref{eq: LL-stochastic}) corresponds to (\ref{eq: LL-classic})
\cite{Landau-Lifshitz}.
\end{lem}
\noindent \end{leftbar}\end{snugshade}

Larmor's radiation formula is included in $\mathscr{P}(\varOmega_{\tau}^{\mathrm{ave}})\times[-eF_{\mathrm{LL}}^{\mu\nu}(\mathbb{\mathbb{E}}\llbracket\hat{x}(\tau,\bullet)\rrbracket)]\cdot d\mathbb{\mathbb{E}}\llbracket\hat{x}_{\nu}(\tau,\bullet)\rrbracket/d\tau]$\footnote{Larmor's formula is found as the coefficient of the ''direct radiation
term'' which is proportional to the velocity $d\mathbb{\mathbb{E}}\llbracket\hat{x}(\tau,\bullet)\rrbracket/d\tau$,
i.e., the term of $-dW/d\tau\times d\mathbb{\mathbb{E}}\llbracket\hat{x}(\tau,\bullet)\rrbracket/d\tau$. }:
\begin{equation}
\frac{dW}{dt}(\mathbb{\mathbb{E}}\llbracket\hat{x}(\tau,\bullet)\rrbracket)=-\mathscr{P}(\varOmega_{\tau}^{\mathrm{ave}})\times\frac{\tau_{0}}{m_{0}}f_{\alpha}(\mathbb{\mathbb{E}}\llbracket\hat{x}(\tau,\bullet)\rrbracket)\cdot f^{\alpha}(\mathbb{\mathbb{E}}\llbracket\hat{x}(\tau,\bullet)\rrbracket)+O(\tau_{0}^{2},\mathbb{E}\llbracket\overset{2}{\otimes}\delta\hat{x}(\tau,\bullet)\rrbracket)
\end{equation}
Where,
\begin{eqnarray}
f^{\mu}(\mathbb{\mathbb{E}}\llbracket\hat{x}(\tau,\bullet)\rrbracket) & = & -eF_{\mathrm{ex}}^{\mu\nu}(\mathbb{\mathbb{E}}\llbracket\hat{x}(\tau,\bullet)\rrbracket)\frac{d\mathbb{\mathbb{E}}\llbracket\hat{x}_{\nu}(\tau,\bullet)\rrbracket}{d\tau}\,.
\end{eqnarray}
Then, let us recall the classical radiation formula $dW_{\mathrm{classical}}/dt=-\tau_{0}/m_{0}\times f_{\alpha}(\mathbb{\mathbb{E}}\llbracket\hat{x}(\tau,\bullet)\rrbracket)\cdot f^{\alpha}(\mathbb{\mathbb{E}}\llbracket\hat{x}(\tau,\bullet)\rrbracket)+O(\tau_{0}^{2})$
(an acceleration $dv/d\tau$ in (\ref{eq: dW_classical/dt}) is perturbed
by the parameter of $\tau_{0}$. see also \cite{Landau-Lifshitz}),
simplify the above formula of $dW/dt$.\begin{snugshade}
\begin{thm}[Radiation formula for a D-process]
\begin{leftbar}The radiation formula for a stochastic scalar electron
is expressed as follows:
\begin{eqnarray}
\frac{dW}{dt}(\mathbb{\mathbb{E}}\llbracket\hat{x}(\tau,\bullet)\rrbracket) & = & \int_{\omega\in\varOmega_{\tau}^{\mathrm{ave}}}d\mathscr{P}(\omega)\frac{dW_{\mathrm{classical}}}{dt}(\hat{x}(\tau,\bullet))\nonumber \\
 & = & \mathscr{P}(\varOmega_{\tau}^{\mathrm{ave}})\times\frac{dW_{\mathrm{classical}}}{dt}(\mathbb{E}\llbracket\hat{x}(\tau,\bullet)\rrbracket)
\end{eqnarray}
\end{leftbar}
\end{thm}
\end{snugshade}Hereby, we can say $q(\chi)$ of the factor in (\ref{eq: Q=000026C Larmor})
is regarded as $\mathscr{P}(\varOmega_{\tau}^{\mathrm{ave}})$ the
probability at the expectation of the particle position.

However in the real practices, the analytical derivation of $\mathscr{P}(\varOmega_{\tau}^{\mathrm{ave}})$
by solving (\ref{eq: prob_evo}) is complicated. Therefore, we assume
the simplified condition for the first experiment on high-intensity
field physics \cite{RA5}. That is just an irradiation of a highly
energetic electron by a long-focal and high-intensity laser. If we
assume it, a plane wave condition becomes feasible. The radiation
formula by the non-linear Compton scattering (the single photon emission
process) of a scalar electron \cite{A. Sokolov(1986)} is derived
like 
\begin{equation}
\frac{dW}{dt}=q_{scalar}(\chi)\times\frac{dW_{\mathrm{classical}}}{dt}\,,
\end{equation}
\begin{equation}
q_{\mathrm{scalar}}(\chi)=\frac{9\sqrt{3}}{8\pi}\int_{0}^{\chi^{-1}}dr\,r\int_{\frac{r}{1-\chi r}}^{\infty}dr'\,K_{5/3}(r')\,,\label{eq: Xi}
\end{equation}
with the definition of the non-linearity parameter,
\begin{equation}
\chi\coloneqq\frac{3}{2}\frac{\hbar}{m_{0}^{2}c^{3}}\sqrt{-f_{\alpha}(\mathbb{\mathbb{E}}\llbracket\hat{x}(\tau,\bullet)\rrbracket)\cdot f^{\alpha}(\mathbb{\mathbb{E}}\llbracket\hat{x}(\tau,\bullet)\rrbracket)}\,.
\end{equation}
By putting the difference $\delta\mathscr{P}(\varOmega_{\tau}^{\mathrm{ave}})$
between $\mathscr{P}(\varOmega_{\tau}^{\mathrm{ave}})$ and $q_{\mathrm{scalar}}(\chi)$
in general,
\begin{equation}
\mathscr{P}(\varOmega_{\tau}^{\mathrm{ave}})\coloneqq q_{\mathrm{scalar}}(\chi)+\delta\mathscr{P}(\varOmega_{\tau}^{\mathrm{ave}})\,.
\end{equation}
Since $r$ is characterized by the relation $r\coloneqq\chi^{-1}\times\hbar\omega/m_{0}c^{2}\gamma$,
the radiation spectrum can be introduced as follows (see also \cite{A. Sokolov(1986),I. Sokolov(2011a),Landau-Lifshitz}):
\begin{equation}
\frac{d^{2}W}{dtd(\hbar\omega)}=-\frac{\tau_{0}}{m_{0}}g_{\mu\nu}f_{\mathrm{Re}}^{\mu}(\mathbb{\mathbb{E}}\llbracket\hat{x}(\tau,\bullet)\rrbracket)f_{\mathrm{Re}}^{\nu}(\mathbb{\mathbb{E}}\llbracket\hat{x}(\tau,\bullet)\rrbracket)\times\left[\frac{dq_{\mathrm{scalar}}(\chi)}{d(\hbar\omega)}+\frac{d\delta\mathscr{P}(\varOmega_{\tau}^{\mathrm{ave}})}{d(\hbar\omega)}\right]
\end{equation}
 
\begin{equation}
\frac{dq_{\mathrm{scalar}}(\chi)}{d(\hbar\omega)}=\frac{9\sqrt{3}}{8\pi\chi^{2}}\frac{\hbar\omega}{(m_{0}c^{2}\gamma)^{2}}\int_{\frac{r}{1-\chi r}}^{\infty}dr'\,K_{5/3}(r')
\end{equation}
So, the existence of $\delta\mathscr{P}(\varOmega_{\tau}^{\mathrm{ave}})$
is the difference from the non-linear Compton scatterings under the
plane wave condition. Thus, plotting $\mathscr{P}(\varOmega_{\tau}^{\mathrm{ave}})$
like {\bf Figure \ref{q_chi}} is interesting in the real experiments
by a high-intensity laser \cite{RA5}.

As the further works, the deeper analysis of this method and numerical
simulations have to be expected to innovate high-intensity field physics
toward together with real experiments carried out by the state-of-the-arts
$O(10\mathrm{PW})$ lasers \cite{ELI-NP,ELI-beams,ELI-ALPS,RA5}.

\section*{\addcontentsline{toc}{section}{Acknowledgements}Acknowledgments}

This work is supported by Extreme Light Infrastructure \textendash{}
Nuclear Physics (ELI-NP) \textendash{} Phase I, and also Phase II,
a project co-financed by the Romanian Government and the European
Union through the European Regional Development Fund.


\begin{thebibliography}{10}
\bibitem{Seto Volume I} K. Seto, ``A Brownian Particle and Fields
I: Construction of Kinematics and Dynamics ({\bf Volume I})'' (2016).

\bibitem{Mourou(1997a)}G. A. Mourou, C. Barty, and M. D. Perry, Phys.
Today {\bf 51}, 22 (1997).

\bibitem{ELI-NP} ELI-NP: \url{https://www.eli-np.ro/}.

\bibitem{ELI-beams} ELI-beamlines: \url{http://www.eli-beams.eu/}.

\bibitem{ELI-ALPS} ELI-ALPS: \url{http://www.eli-hu.hu/}.

\bibitem{Piazza(2012a)} A. Di Piazza, C. M\"{u}ller, K. Z. Hatsagortsyan,
and C. H. Keiel, Rev. Mod. Phys. {\bf 84}, 1177 (2012).

\bibitem{Larmor} J. Larmor, Phil. Trans. Roy. Soc. London A {\bf 190}
205 (1897).

\bibitem{Dirac(1938a)} P. A. M. Dirac, Proc. Roy. Soc. A {\bf 167},
148 (1938).

\bibitem{Seto(2015a)} K. Seto, Prog. Theor. Exp. Phys., {\bf 2015},
103A01 (2015).

\bibitem{I. Sokolov(2011a)}  I. V. Sokolov, N. M. Naumova, and J.
A. Nees, Phys. Plasmas {\bf 18}, 093109 (2011).

\bibitem{Volkov(1935)} D. M. Volkov, Z. Phys {\bf 94}, 250 (1935).

\bibitem{Bergou(1980)} J. Bergou, and S. Varro, J. Phys. A: Math.
Gen. {\bf 12}, 2823 (1980).

\bibitem{Zakowicz(2005)} S. Zakowicz, J. Math. Phys. {\bf 46}, 032304
(2005).

\bibitem{Boca-Florescu(2010)} M. Boca, and V. Florescu, Rom. J. Phys.
{\bf 55}, 511 (2010).

\bibitem{Berestetskii-Lifshitz QED} V. B. Berestetskii, E. M. Lifshitz,
and L. P. Pitaevskii, {\it Quantum Electrodynamics} (Elsevier, Oxford,
1982).

\bibitem{Furry(1951)} W. H. Furry, Phys. Rev. {\bf 85}, 115 (1951).

\bibitem{Ritus(1972a)} V. I. Ritus, Ann. Phys. {\bf 69}, 555 (1972).

\bibitem{Brown-Kibble(1964)} L.L. Brown, and T. W. B. Kibble, Phys
Rev. {\bf 133}, A705 (1964).

\bibitem{Nikishov(1964a))} A. I. Nikishov, and V. I. Ritus, Zh. Eksp.
Teor. Fiz. {\bf 46}, 776 (1963) {[}Sov. Phys. JETP {\bf 19}, 529
(1964){]}.

\bibitem{Nikishov(1964b)} A. I. Nikishov, and V. I. Ritus, Zh. Eksp.
Teor. Fiz. {\bf 46}, 1768 (1964) {[}Sov. Phys. JETP {\bf 19}, 1191
(1964){]}.

\bibitem{A. Sokolov(1986)} A. A. Sokolov, and I. M. Ternov, {\it "Radiation from Relativistic Electrons"},
(American Institute of Physics, transration series, 1986).

\bibitem{Koga(2005)} J. Koga, T. Z. Esirkepov, and S. V. Bulanov,
Phys. Plasmas {\bf 12}, 093106 (2005).

\bibitem{Harvey(2009)} C. Harvey, T. Heinzl, and A. Ilderton, Phys.
Rev. A {\bf 79}, 063407 (2009).

\bibitem{Mackenroth(2011)} F. Mackenroth, and A. Di Piazza, Phys.
Rev. A {\bf 83}, 032106 (2011).

\bibitem{Neitz(2013)} N. Neitz, and A. Di Piazzam Phys. Rev. Lett.
{\bf 111}, 054802 (2013).

\bibitem{Seipt(2013)} D. Seipt, and B. K\"{a}mpfer, Phys. Rev. A
{\bf 88}, 012127 (2013).

\bibitem{Anton(2013)} A. Ilderton, and G. Torgrimsson, Phys. Rev.
D {\bf88}, 025021 (2013). 

\bibitem{RA5} K. Homma, O. Tesileanu, L.D'Alessi, T. Hasebe, A. Ilderton,
T. Moritaka, Y. Nakamiya, K. Seto, and H. Utsunomiya, Rom. Rep.Phys.
{\bf 68}, Supplement, S233 (2016).

\bibitem{Roos(1966)} O. V. Roos, Phys. Rev. {\bf 150}, 1112 (1966).

\bibitem{Fried(1966)} Z. Fried, A. Baker, and D. Korff, Phys. Rev.
{\bf 151}, 1040 (1966).

\bibitem{Guo(1991)} D-S. Guo, and T \r{A}berg, J. Phys. B: At. Mol.
Opt. Phys. {\bf 24}, 349 (1991).

\bibitem{Piazza} A. Di Piazza, Phys. Rev. A {\bf 95}, 032121 (2017). 

\bibitem{SLAC} C. Bula, K. T. McDonald, E. J. Prebys, C. Bamber,
S. Boege, T. Kotseroglou, A. C. Melissinos, D. D. Meyerhofer, W. Ragg,
D. L. Burke, R. C. Field, G. Horton-Smith, A. C. Odian, J. E. Spencer,
and D. Walz, Phys. Rev. Lett. {\bf 76}, 3116 (1996); C. Bamber, S.
J. Boege, T. T. Koffas, T. Kotseroglou, A. C. Melissinos, D. D. Meyerhofer,
D. A. Reis, W. Raggi, C. Bula, K. T. McDonald, E. J. Prebys, D. L.
Burke, R. C. Field, G. Horton-Smith, J. E. Spencer, D. Walz, S. C.
Berridge, W. M. Bugg, K. Shmakov, and A. W. Weidemann, Phys. Rev.
D {\bf 60}, 092004 (1999).

\bibitem{Nelson(1966a)} E. Nelson, Phys. Rev. {\bf 150}, 1079 (1966).

\bibitem{Nelson(2001_book)} E. Nelson, {\it "Dynamical Theory of Brownian Motion"},
(Princeton University Press, 2nd Ed., 2001)

\bibitem{Hartemann(2001)} F. V. Hartemann, {\it "High-Field Electrodynamics"}
(CRC Press, 2001).

\bibitem{explanation1} This depends on the selection of the initial
condition. Namely, the retarded field $\mathcal{F}_{(\mathring{-})}$
is neglected when we get the initial condition at $\tau=+\infty$.
Then, $C=0$ and the dynamics is detected by the time-inversed evolution
from $\tau=+\infty$.

\bibitem{Ford-O'Connell} G. W. Ford, and R. F. O'Connell, Phys. Lett.
A {\bf 174}, 182 (1993).

\bibitem{Landau-Lifshitz} L. D. Landau, and E. M. Lifshitz, {\it "The Classical Theory of Fields"}
(Pergamon, New York, 1994).\end{thebibliography}
\end{document}